\documentclass[runningheads]{llncs_thmcountersep_modified}
\usepackage{amssymb}
\usepackage{amsmath}
\usepackage{graphicx}
\usepackage[ruled,linesnumbered,noend]{algorithm2e}
\usepackage{comment}
\usepackage{tabu}
\usepackage{multirow}
\usepackage{cite}
\usepackage{float}
\usepackage{wrapfig}
\usepackage[table]{xcolor}

\usepackage[colorinlistoftodos,textsize=footnotesize]{todonotes}

\usepackage{tikz}
\usetikzlibrary{automata, positioning, arrows, shapes.geometric}
\tikzset{
->, 
>=stealth', 
node distance=2cm, 
every state/.style={thick}, 
initial text=$ $, 
}

\usepackage{theorem}
\theoremstyle{plain}
\newtheorem{mytheorem}{Theorem}[section]
\newtheorem{myproposition}[mytheorem]{Proposition}
\newtheorem{mylemma}[mytheorem]{Lemma}
\theorembodyfont{\normalfont}
\theoremstyle{definition}
\newtheorem{myexample}[mytheorem]{Example}

\newtheorem{mydefinition}[mytheorem]{Definition}

\newcommand\numberthis{\addtocounter{equation}{1}\tag{\theequation}}

\usepackage{array}
\newcolumntype{C}[1]{>{\centering\let\newline\\\arraybackslash\hspace{0pt}}m{#1}}
\newcolumntype{R}[1]{>{\raggedleft\let\newline\\\arraybackslash\hspace{0pt}}m{#1}}
\newcolumntype{L}[1]{>{\raggedright\let\newline\\\arraybackslash\hspace{0pt}}m{#1}}

\newcommand{\game}{\mathcal{G}}
\newcommand{\mdp}{\mathcal{M}}
\newcommand{\wg}{\mathcal{W}}
\newcommand{\Av}{\mathrm{A\kern-0.09em v}}
\newcommand{\post}{\mathrm{post}}
\newcommand{\nat}{\mathbb{N}}
\newcommand{\mynext}{\mathbb{X}}
\newcommand{\prob}{\mathbb{P}}

\newcommand{\zero}{\mathbf{0}}
\newcommand{\one}{\mathbf{1}}
\newcommand{\str}{\mathrm{str}}
\newcommand{\minimizer}{\bigcirc}
\newcommand{\maximizer}{\square}
\newcommand{\playerReduction}{\mdp_\mathrm{PlRd}}
\newcommand{\localPropagation}{\wg_\mathrm{LcPg}}
\newcommand{\wpw}{\mathrm{WPW}}
\newcommand{\place}{\underline{\phantom{n}}\,} 

\DeclareMathOperator*{\argmin}{arg\,min}

\newcommand{\wpath}{{\rm WPath}}

\usepackage{subfig}
\usepackage{stfloats}

\usepackage[all,pdftex]{xy}
\SelectTips{cm}{}
\xyoption{rotate}

\begin{document}
\title{Widest Paths and Global Propagation in Bounded Value Iteration for Stochastic Games 
	}
%
%
\author{
Kittiphon Phalakarn\inst{1}\thanks{The work was done during K.P.'s internship at National Institute of Informatics, Japan, while he was a student at Chulalongkorn University, Thailand.}
\and
Toru Takisaka\inst{2}
 \and
Thomas Haas\inst{3} \and
Ichiro Hasuo\inst{2,4}
}
\authorrunning{
K.\ Phalakarn et al.
}
%
\institute{
University of Waterloo, Waterloo, Canada \\ \email{kphalakarn@uwaterloo.ca} \and
National Institute of Informatics, Tokyo, Japan \\ \email{\{takisaka,hasuo\}@nii.ac.jp}\and
Technical University of Braunschweig, Braunschweig, Germany \\ \email{thohaas@tu-bs.de}\and
The Graduate University for Advanced Studies (SOKENDAI), Tokyo, Japan 
}
\maketitle              
\begin{abstract}
Solving \emph{stochastic games} with the reachability objective is a fundamental problem, especially in quantitative verification and synthesis. For this purpose, \emph{bounded value iteration (BVI)} attracts attention as an efficient iterative method. However, BVI's performance is often impeded by costly \emph{end component (EC) computation} that is needed to ensure convergence. Our contribution is a novel BVI algorithm that conducts, in addition to local propagation by the Bellman update that is typical of BVI, \emph{global} propagation of upper bounds that is not hindered by ECs. To conduct global propagation in a computationally tractable manner, we construct a weighted graph and solve the \emph{widest path problem} in it. Our experiments show the algorithm's performance advantage over the previous BVI algorithms that rely on EC computation.

\end{abstract}
%
%
%

\section{Introduction} 
\subsection{Stochastic Game (SG)}
A \emph{stochastic game}~\cite{Condon92} is a two-player game played on a graph. In an SG, an action $a$ of a player causes a transition from the current state $s$ to a successor $s'$, with the latter chosen from a prescribed probability distribution $\delta(s,a,s')$. Under the reachability objective, the two players (called \emph{Maximizer} and \emph{Minimizer}) aim to maximize and minimize,  respectively, the reachability probability to a designated target state. 

Stochastic games are a fundamental construct in theoretical computer science, especially in the analysis of probabilistic systems. Its complexity is interesting in its own: the problem of threshold reachability---whether Maximizer has a strategy that ensures the reachability probability to be at least given $p$---is known to be in $\mathsf{UP}\cap\mathsf{coUP}$~\cite{HoffmanK66}, but no polynomial algorithm is known. The practical significance of SGs comes from the number of problems that can be encoded to SGs and then solved. Examples include the following: solving deterministic parity games~\cite{ChatterjeeF11}, solving stochastic games with the parity or mean-payoff objective~\cite{AnderssonM09}, and a variety of probabilistic verification and reactive synthesis problems in different application domains such as cyber-physical systems. See e.g.~\cite{SvorenovaK16}.

SGs are often called 2.5-player games, where probabilistic branching is counted as 0.5 players. They generalize deterministic automata (0-player), Markov chains (MCs, 0.5-player), nondeterministic automata (1-player), Markov decision processes (MDPs, 1.5-player) and (deterministic) games (2-player). Many theoretical considerations on these special cases carry over smoothly to SGs. However, SGs have their peculiarities, too. One example is the treatment of end components in bounded value iteration, as we describe later.

\subsection{Value Iteration (VI) }
In an SG, we are interested in the \emph{optimal} reachability probability, that is, the reachability probability when both Maximizer and Minimizer take their optimal strategies. The function that returns these optimal reachability probabilities is called the \emph{value function} $V(\game)$ of the SG $\game$; our interest is in computing this value function, desirably constructing optimal strategies for the two players at the same time. For this purpose, two principal families of solution methods are \emph{strategy iteration (SI)}~\cite{HoffmanK66} and \emph{value iteration (VI)}~\cite{Condon92,ChatterjeeH08}---the latter is commonly preferred for performance reasons.

\begin{wrapfigure}[5]{r}{0.18\textwidth}
\vspace*{-2em}

\begin{math}
  \vcenter{\xymatrix@R=-.3em@C-1em{
  &
  *+<1pc>[o][F-]{s_{1}}
  \\
 *+<1pc>[o][F-]{s}
   \ar[ru]^-{p_{1}}
   \ar[rd]_-{p_{n}}  
  &
  \vdots
  \\
  &
  *+<1pc>[o][F-]{s_{n}}
 }}
\end{math}
\end{wrapfigure}
The mathematical principle that underpins VI is the characterization of the value function $V(\game)$ as the \emph{least fixed point (lfp)} of an function update operator $\mynext$ called  the \emph{Bellman operator}. The Bellman operator $\mynext$  back-propagates function values  by one step, using the average. For the simple case of Markov chains shown on the right, it is defined by $(\mynext f)(s)=\sum_{i} p_{i}\cdot f(s_{i})$, turning a function $f\colon S\to [0,1]$ (i.e., assignment of ``scores'' to states) to $\mynext f\colon S\to [0,1]$.

Since $V(\game)$ is the lfp $\mu\mynext$, Kleene's fixed point theorem tells us the sequence 
 \begin{equation}\label{eq:kleeneChainForValueFuncIntro}
  \bot \;\le\; \mynext\bot \;\le\; \mynext^{2}\bot \;\le\; \cdots,
 \end{equation} 
where $\bot$ is the least element of the function space $S\to [0,1]$, converges to $V(\game)=\mu\mynext$. VI consists of the iterative approximation of $V(\game)$ via the sequence~(\ref{eq:kleeneChainForValueFuncIntro}). 

An issue from the practical point of view, however, is that $\mynext^{i}\bot$ never becomes equal to $V(\game)$ in general. Even worse, one cannot know how close the current approximant $\mynext^{i}\bot$ is to the desired function $V(\game)$~\cite{HADDAD2018111}. In summary, VI as an iterative approximation method does not give any precision guarantee.

\subsection{Bounded Value Iteration (BVI) and End Components}
Bounded value iteration (BVI) has been actively studied as an extension of VI that comes with a precision guarantee~\cite{McMahanLG05,Brazdil2014verification,HADDAD2018111,kelmendi2018value,AshokKW19,EisentrautKA19preprint,AshokKW19preprint}. Its core ideas are the following two. 

\begin{wrapfigure}[4]{r}{0.35\textwidth}

\begin{math}
  \vcenter{\xymatrix@R=-.5em@C=1em{
 L_{0}
   \ar@{}[r]|{\le}
 &
 L_{1}
   \ar@{}[r]|{\le}
 &
 \cdots
   \ar@{}[rd]|*[@rd]{\le}
 \\
 &&&
 V(\game)
 \\
 U_{0}
   \ar@{}[r]|{\ge}
 &
 U_{1}
   \ar@{}[r]|{\ge}
 &
 \cdots
   \ar@{}[ru]|*[@ru]{\ge}
 }}
\end{math}
 \hfill
\parbox[c]{0em}{\begin{equation}\label{eq:LandUIntro}
\end{equation}
}\end{wrapfigure}
Firstly, BVI computes not only iterative lower bounds $L_{i}=\mynext^{i}\bot$ for $V(\game)$, but also iterative \emph{upper bounds} $U_{i}$, as shown on the right in~(\ref{eq:LandUIntro}). This gives us a precision guarantee---$V(\game)$ must lie between the approximants $L_{i}$  and $U_{i}$. 

Secondly, for computing upper bounds $U_{i}$, BVI uses the Bellman operator again: $U_{i}=\mynext^{i}\top$ where $\top$ is the greatest element of the  function space $S\to [0,1]$. This leads to the following approximation sequence that is dual to~(\ref{eq:kleeneChainForValueFuncIntro}):
\begin{equation}\label{eq:kleeneChainForValueFuncIntroFromTop}
   \top \;\ge\; \mynext\top \;\ge\; \mynext^{2}\top \;\ge\; \cdots.
\end{equation}
The sequence~(\ref{eq:kleeneChainForValueFuncIntroFromTop}) converges to the \emph{greatest fixed point (gfp)} $\nu \mynext$ of $\mynext$, which must be above the lfp $V(\game)=\mu\mynext$. Therefore the elements in~(\ref{eq:kleeneChainForValueFuncIntroFromTop}) are all above $V(\game)$. 

The problem, however, is that the gfp $\nu \mynext$ is not necessarily the same as $\mu\mynext$. Therefore the upper bounds $U_{0}\ge U_{1}\ge\cdots$ given by~(\ref{eq:kleeneChainForValueFuncIntroFromTop}) may not converge to $V(\game)$. In other words, for a given threshold $\varepsilon >0$, the bounds in~(\ref{eq:LandUIntro}) may fail to achieve $U_{i}-L_{i}\le \varepsilon$, no matter how large  $i$ is.

\begin{wrapfigure}[7]{r}{0.33\textwidth}
\vspace{-3em}

\centering
 \begin{math}
   \xymatrix@1@C-1.5em{%
*+<1pc>[o][F-]{s_{I}}
\ar@/_1pc/[r]_{1}
 &
*+<1pc>[o][F-]{s}
\ar@/_1pc/[l]_{1}
& *+<1pc>[o][F-]{\one} 
& *+<1pc>[o][F-]{\zero} 
}
 \end{math}
\caption{A Markov chain (MC) for which the naive BVI fails to converge}
\label{fig:AMarkovchainforwhichthenaiveBVIfailstoconverge}
\end{wrapfigure}
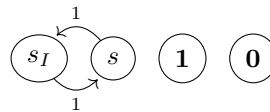
In the literature, the source of this convergence issue has been identified to be \emph{end components (ECs)} in MCs/MDPs/SGs. 
ECs are much like loops without exits---an example is in Fig.~\ref{fig:AMarkovchainforwhichthenaiveBVIfailstoconverge}, where we use a Markov chain (MC) for simplicity.  Any function $f$ that assigns the same value to the states $s_{I}$ and $s$ can be a fixed point of the Bellman operator $\mynext$ (that back-propagates $f$ by averages); therefore, the gfp $\nu \mynext$ assigns $1$ to both $s_{I}$ and $s$. In contrast, $(\mu\mynext)(s_{I})=(\mu\mynext)(s)=0$, which says one never reaches the target $\one$ from $s_{I}$ or $s$ (which is obvious).

Most previous works on BVI have focused on the problem of how to deal with ECs. Their solutions are to get somehow rid of ECs. For example, ECs in MDPs are discovered and \emph{collapsed} in~\cite{Brazdil2014verification,HADDAD2018111}; ECs in SGs cannot simply be collapsed, and an elaborate method is proposed in the recent work~\cite{kelmendi2018value} that \emph{deflates} them. This is the context of the current work, and we aim to enhance BVI for SGs.

\subsection{Contribution: Global Propagation in BVI with Widest Paths}
The algorithms in~\cite{kelmendi2018value} seem to be the only BVI algorithms known for SGs. In their performance, however,  EC computation often becomes a bottleneck. 
Our contribution in this paper is a new BVI algorithm for SGs that is \emph{free from the need for EC computation}. 

The key idea of our algorithm is \emph{global propagation} for upper bounds, as sketched below. In each iteration for upper bounds $U_{0}\ge U_{1}\ge\cdots$, we conduct \emph{global} propagation, in addition to the \emph{local} propagation in the usual BVI. The latter means the application of $\mynext$ to $\mynext^{i}\top$, leading to $\mynext^{i+1}\top$; this local propagation, as we previously discussed, gets trapped in end components. In contrast, our global propagation looks at paths from each state $s$ to the target $\one$,  ignoring end components. For example, in Fig.~\ref{fig:AMarkovchainforwhichthenaiveBVIfailstoconverge}, our global propagation  sees that there is no path from $s_{I}$ to the target $\one$, and consequently assigns $0$ as an upper bound for the value function $V(\game)(s_{I})$. 

Such global propagation is easier said than done---in fact, the very advantage of VI is that the \emph{global} quantities (namely reachability probabilities) get computed by iterations of \emph{local} propagation. Conducting global propagation in a computationally tractable manner requires a careful choice of its venue. The solution in this paper is to compute \emph{widest paths} in a suitable (directed) weighted graph. 

More specifically, in each iteration where we compute an upper bound $U_{i}$, we conduct the following operations.
\begin{itemize}
 \item (Player reduction)
 We  turn the given SG $\game$  into an MDP $\mdp_{i}$, by restricting Minimizer's actions  to the \emph{empirically optimal} ones. The latter means they are optimal with respect to the current under-approximation $L_{i}$ of  $V(\game)$. 
 \item (Local propagation)
 The MDP $\mdp_{i}$ is then turned into a weighted graph (WG) $\wg_{i}$. The construction of $\wg_{i}$ consists of the application of $\mynext$ to the previous bound $U_{i-1}$ (i.e.\ local propagation), and forgetting the information that cannot be expressed in a weighted graph (such as the precise transition probabilities $\delta(s,a,s')$ that depend on the action $a$). 

Due to this information loss, our analysis in $\wg_{i}$ is necessarily approximate. Nevertheless, the benefit of $\wg_{i}$'s simplicity is significant, as in the following step.
 \item (Global propagation)
In the WG $\wg_{i}$, we solve the \emph{widest path problem}. This classic graph-theoretic problem can be solved efficiently, e.g., by the Dijkstra algorithm.  The widest path width gives a new upper bound $U_{i}$.
\end{itemize}

We prove the correctness of our algorithm: soundness ($V(\game)\le U_{i}$), and convergence ($U_{i}\to V(\game)$ as $i\to\infty$). That the upper bounds decrease ($U_{0}\ge U_{1}\ge\cdots$) will be obvious by construction. These correctness proofs are technically nontrivial, combining combinatorial, graph-theoretic, and  analytic arguments. 

We have also implemented our algorithm. Our experiments compare its performance to the algorithms from~\cite{kelmendi2018value} (the original one and its learning-based variation). The results show our consistent performance advantage: depending on SGs, our performance is from comparable to dozens of times faster. The advantage is especially eminent in SGs with many ECs.

\subsection{Related Works}
VI and BVI have been pursued principally for MDPs. The only work we know that deals with SGs is~\cite{kelmendi2018value}---with the exception of~\cite{Ujma15} that works in a restricted setting where every end component belongs exclusively to either player. The work closest to ours is therefore~\cite{kelmendi2018value}, in that we solve the same problem. 

For MDPs, the idea of BVI is first introduced in~\cite{McMahanLG05}; they worked in a limited setting where ECs do not cause the convergence issue. Its extension to general MDPs with the reachability objective is presented in~\cite{Brazdil2014verification,HADDAD2018111}, where ECs are computed and then collapsed. BVI is studied under different names in these works: \emph{bounded real time dynamic programming}~\cite{McMahanLG05,Brazdil2014verification} and \emph{interval iteration}~\cite{HADDAD2018111}. 
 The work~\cite{kelmendi2018value} is an extension of this line of work from MDPs to SGs. 

The work~\cite{kelmendi2018value} has  seen a few extensions to more advanced settings: black-box settings~\cite{AshokKW19}, concurrent reachability~\cite{EisentrautKA19preprint}, and generalized reachability games~\cite{AshokKW19preprint}.

Most BVI algorithms involve EC computation (although ours does not). 
The EC algorithm in~\cite{DeAlfaro:1998:FVP:927475,CourcoubetisY95} is used  in~\cite{kelmendi2018value,HADDAD2018111}; more recent algorithms include~\cite{ChatterjeeDHS19,chatterjee2014efficient}.

\subsection{Organization} 
In~\S{}\ref{sec:preliminaries} we present some preliminaries.
In~\S{}\ref{sec:VIAndBVI} we review 
VI and BVI 
with an emphasis on the role of Kleene's fixed point theorem. This paves the way to~\S{}\ref{sec:ourAlgorithm} where we present our algorithm. We do so in three steps, and prove the correctness---soundness and convergence---in the end. Experiment results are shown in~\S{}\ref{sec:experiment}. 


\section{Preliminaries}\label{sec:preliminaries}
We fix some basic notations. Let $X$ be a set. We let $X^{*}$ denote the set of finite sequences over $X$, that is, $X^{*}=\bigcup_{i\in\nat}X^{i}$. We let $X^{+}=X^{*}\setminus\{\varepsilon\}$, where $\varepsilon$ denotes the empty sequence (of length $0$). The set of infinite sequences over $X$ is denoted by $X^{\omega}$.  The set of functions from $X$ to $Y$ is denoted by  $X\to Y$.

\subsection{Stochastic Games}
In a stochastic game, two players (\emph{Maximizer} $\square$ and \emph{Minimizer} $\bigcirc$) play against each other.
The goals of the two players are to maximize and minimize the  \emph{value function}, respectively. 
Many different definitions are possible for value functions. In this paper (as well as all the works on (bounded) value iteration), we focus on the \emph{reachability objective}, in which case a value function is defined by the reachability probability to a designated target state $\one$. 
\begin{mydefinition}[stochastic game (SG)] \label{def:SG}
	A \textnormal{stochastic game (SG)} is a tuple 
	$\game= (S, S_{\square}, S_{\bigcirc}, s_I,$ 
	$ \mathbf{1}, \mathbf{0}, A, \Av, \delta)$ 
        where
	\begin{itemize}
		\item $S$ is a finite set of \emph{states}, partitioned into $S_{\maximizer}$ and $S_{\minimizer}$ (i.e., $S = S_{\square} \cup S_{\bigcirc}$, $S_{\square} \cap S_{\bigcirc} = \emptyset$).  $s\in S_{\maximizer}$ is \emph{Maximizer's} state; $s\in S_{\minimizer}$ is \emph{Minimizer's} state.
		\item $s_I \in S$ is an \emph{initial} state, $\mathbf{1} \in S_{\square}$ is a \emph{target}, and $\mathbf{0} \in S_{\bigcirc}$ is a \textnormal{sink}. 
		\item $A$ is a finite set of \emph{actions}.
		\item $\Av : S \rightarrow 2^{A}$ defines the set of actions that are \emph{available} at each state $s\in S$. 
		\item $\delta : S \times A \times S \rightarrow [0,1]$ is a \emph{transition function}, where $\delta(s,a,s')$ gives a probability with which to reach the state $s'$ when the action $a$ is taken at the state $s$. The value $\delta(s,a,s')$  is non-zero only if $a \in \Av(s)$; it  must  satisfy $\sum_{s' \in S} \delta(s,a,s') = 1$ for all $s \in S$ and $a \in \Av(s)$.
	\end{itemize}
 We assume that each of $\one$ and $\zero$ allows only one action that leads to a self-loop with probability $1$. 
Moreover, 
for theoretical convenience, we assume that all SGs are non-blocking. That is, $\Av(s)\neq\emptyset$ for each $s\in S$. 

We introduce some notations: $\post(s,a) = \{s' \mid \delta(s
,a,s') > 0\}$, and for $S' \subseteq S$, we  let $S'_{\square} = S' \cap S_{\square}$ and $S'_{\bigcirc} = S' \cap S_{\bigcirc}$.

\end{mydefinition}

\begin{mydefinition}[Markov decision process (MDP), Markov chain (MC)]
 An SG such that $S_\maximizer = S \setminus \{\zero\}$ (i.e.\ Minimizer is absent) is called a \emph{Markov decision process (MDP)}. We often omit the second and third components for MDPs, writing $\mdp=(S, 
s_I,
\mathbf{1}, \mathbf{0}, A, \Av, \delta)$.

An SG such that $|\Av(s)| = 1$ for each $s \in S$---both Maximizer and Minimizer are absent---is called a \emph{Markov chain (MC)}. It is also denoted simply by a tuple $\game = (S, s_I, \mathbf{1}, \mathbf{0}, \delta)$ where its transition function is of the type $\delta:S\times S \to [0,1]$. 

Every notion for SGs that appears below applies to MDPs and MCs, too.
\end{mydefinition}

\begin{myexample}
Fig.~\ref{fig:ex_sg} presents an example of an SG. 
At the state $s_1$ of Minimizer, 
two actions $\alpha$ and $\beta$ are in $\Av(s_1)$. 
If Minimizer chooses $\alpha$, the next state is $s_2$ with probability $\delta(s_1,\alpha,s_2) = 1$. 
If Minimizer instead chooses $\beta$, the next state is 
$\mathbf{1}$ with probability $\delta(s_1,\beta,\mathbf{1}) = 0.8$ or $\mathbf{0}$ with probability $\delta(s_1,\beta,\mathbf{0}) = 0.2$. 

Maximizer's goal is to reach $\one$ as often as possible by choosing suitable actions. Minimizer's goal is to avoid reaching $\one$---this can be achieved, for example but not exclusively, by reaching $\zero$.
\end{myexample}

\begin{wrapfigure}[7]{R}{0.55\textwidth}
\vspace*{-3.5em}

 \centering
\scalebox{.7}{    \begin{tikzpicture}[square/.style={regular polygon,regular polygon sides=4}]
        \node[state] at (0,0) (s1) {$s_1$};
        \node[state, square, initial] at (-150:2.5) (s0) {$s_I$};
        \node[state, square] at (0,-2.5) (s2) {$s_2$};
        \node[state, square, minimum size=1cm] at (2.5,0) (target) {$\mathbf{1}$};
        \node[state] at (2.5,-2.5) (sink) {$\mathbf{0}$};
        \draw
        (target) edge[loop right] node[pos=0.3,above right]{$\alpha$} node[pos=0.75,below right]{1} (target)
        (sink)   edge[loop right] node[pos=0.3,above right]{$\alpha$} node[pos=0.75,below right]{1} (sink)
        (s0)     edge[above, bend left=20] node[pos=0.2,above left]{$\alpha$} node[pos=0.75]{1} (s1)
        (s0)     edge[above, bend right=20] node[pos=0.1,right]{$\gamma$} node[pos=0.75,below]{0.6} (s2)
        (s1)     edge node[pos=0.2,right]{$\alpha$} node[pos=0.7,right]{1} (s2)
        (s1)     edge node[pos=0.2,above]{$\beta$} node[pos=0.7,above]{0.8} (target)
        (s2)     edge node[pos=0.2,above]{$\alpha$} node[pos=0.7,above]{0.1} (sink)
        (s2)     edge[bend right=40] node[pos=0.2,below]{$\beta$} node[pos=0.7,below]{1} (sink);
        \draw[->] (-1,-1.25) to node[sloped,above]{0.7} (s1);
        \draw[->] (-1,-1.25) to node[sloped,above]{0.3} (s2);
        \draw[->] (-1.2,-2.18) to[bend left=50] node[below]{0.4} (s0);
        \draw[-] (s0) to[above] node[above]{$\beta$} (-1,-1.25);
        \draw[->] (1,0) to node[sloped,above,pos=0.3]{0.2} (sink);
        \draw[->] (1,-2.5) to node[sloped,above,pos=0.3]{0.9} (target);
    \end{tikzpicture}}
    \caption{A stochastic game (SG), an example}
    \label{fig:ex_sg}
\end{wrapfigure}
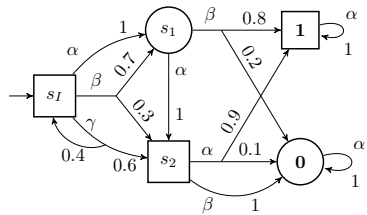

Both players choose their actions according to their \emph{strategies}. It is well-known~\cite{Condon92} that \emph{positional} (also called \emph{memoryless}) and \emph{deterministic} (also called \emph{pure}) strategies are complete for finite SGs with the reachability objective. 
\begin{mydefinition}[strategy, path]\label{def:strategyPath}
Let $\game$ be the SG in Def.~\ref{def:SG}. 
 A  \emph{strategy} for Maximizer in $\game$ is  a function $\sigma\colon S_{\square} \to A$ such that $\sigma(s)\in \Av(s)$ for each $s\in S_{\maximizer}$. A   \emph{strategy} for Minimizer is defined similarly. 
The set of Maximizer's strategies in $\game$ is denoted by
$\str^{\game}_{\maximizer}$; that of Minimizer's is denoted by
 $\str^{\game}_{\minimizer}$.

Strategies  $\tau\in \str^{\game}_{\maximizer}$ and $\sigma\in \str^{\game}_{\minimizer}$ 
 in $\game$ turn the game $\game$ into a Markov chain, which is denoted by $\game^{\tau,\sigma}$. Similarly, a strategy $\tau$ for Maximizer (who is the only player) in an MDP $\mdp$ induces an MC, denoted by $\mdp^{\tau}$.

An \emph{infinite path} in $\game$ is a sequence $s_0 a_0 s_1 a_1 s_2 a_2 \ldots \in (S \times A)^\omega$ such that for all $i \in \nat$, $a_i \in \Av(s_i)$ and $s_{i+1} \in \post(s_i,a_i)$. 
A prefix $s_0 a_0 s_1 \ldots s_k$ of an infinite path ending with a state is called a \emph{finite path}. 
If $\game$ is an MC, then we omit actions in a path and write $s_0s_1s_2\ldots$ or $s_0s_1 \ldots s_k$.

Given a game $\game$ and strategies $\tau,\sigma$ for the two players, the induced MC $\game^{\tau,\sigma}$ assigns to each state $s\in S$ a 
 probability distribution $\mathbb{P}_s^{\tau,\sigma}$. The distribution is with respect to the standard measurable structure of $S^{\omega}$; see, e.g.,~\cite[Chap.~10]{baier2008principles}. For each measurable subset $X \subseteq S^\omega$, $\mathbb{P}_s^{\tau,\sigma}(X)$ is the probability with which
$\game^{\tau,\sigma}$, starting from the state $s$, produces an infinite path $\pi$ that belongs to $X$.
\end{mydefinition}

It is well-known that all the LTL properties are measurable in $S^{\omega}$. In the current setting with the reachability objective, we are interested in the probability of eventually reaching $\one$, denoted by $\mathbb{P}_s^{\tau,\sigma}(\Diamond\one)$. 
\begin{mydefinition}[value function $V(\game)$] \label{def:valueFunc}
Let $\game$ be the SG in Def.~\ref{def:SG}. The \emph{value function} $V(\game)$ of $\game$ is defined by 
\begin{displaymath}
 V(\game)(s) = 
\max_{\tau\in \str^{\game}_{\maximizer}}
\min_{\sigma\in \str^{\game}_{\minimizer}}
\mathbb{P}_s^{\tau,\sigma}(\Diamond\one)
=
\min_{\sigma\in \str^{\game}_{\minimizer}}
\max_{\tau\in \str^{\game}_{\maximizer}}
\mathbb{P}_s^{\tau,\sigma}(\Diamond\one),
\end{displaymath}
where the last equality is shown in~\cite{Condon92}.

We say a strategy $\tau$ of Maximizer's is \emph{optimal} if $V(\game)(s) = \min_\sigma \prob_s^{\tau,\sigma}(\lozenge \one)$ for each $s\in S$; similarly, we say a strategy $\sigma$ of  Minimizer's is \emph{optimal} if $V(\game)(s) = \max_\sigma \prob_s^{\sigma,\tau}(\lozenge \one)$ for each $s \in S$.

We write $V$ for $V({\game})$ when the dependence on $\game$ is clear from the context. 

The set of states with a non-zero value is denoted by $S_{\Diamond\one}$. That is, $S_{\Diamond\one}=\{s\in S\mid V(\game)(s)> 0\}$. 
\end{mydefinition}

\begin{myexample}
Consider the SG $\game$ from Fig.~\ref{fig:ex_sg}. At $s_2$, Maximizer's action should be $\alpha$. Hence, $V(\game)(s_2) = 0.9$. At $s_1$, if Minimizer chooses $\alpha$, then the probability of reaching $\mathbf{1}$ will be $0.9$ by $V(\game)(s_2)$. Thus, Minimizer should choose $\beta$ at $s_1$, which yields $V(\game)(s_1) = 0.8$. Finally, at $s_I$, $\gamma$ is the best choice, since Maximizer can choose this action infinitely often until it gets to $s_2$. We have $V(\game)(s_I) = 0.9$. 
\end{myexample}

\subsection{The Widest Path Problem}\label{subsec:wpp}
\begin{mydefinition}[weighted graph (WG)]\label{def:weightedGraph}
 A (directed) \emph{weighted graph} is a triple $\wg=(V,E,w)$ of a finite set $V$ of \emph{vertices}, a set $E\subseteq V\times V$ of \emph{edges}, and a \emph{weight function} $w\colon E\to [0,1]$ where $[0,1]$ is the unit interval.

 A (finite) \emph{path} in a WG is defined in the usual graph-theoretic way. 
\end{mydefinition}

In the widest path problem, an edge weight $w(v,v')$ is thought of as its capacity, and the capacity of a path is determined by its bottleneck. The problem asks for a path with the greatest capacity.
In this paper, we use the following \emph{all-source single-destination} version of the problem. 
\begin{mydefinition}[the widest path problem (WPP)]\label{def:wpp}
A (finite) \emph{path} in $\wg=(V,E,w)$ is a sequence $v_{0}v_{1}\dotsc v_{n}$ of vertices such that $(v_{i},v_{i+1})\in E$ for each $i\in [0,n-1]$. The \emph{width} of a path $v_{0}v_{1}\dotsc v_{n}$  is given by $\min_{i\in [0,n-1]}w(v_{i},v_{i+1})$. 

The \emph{widest path problem} is the following problem.
\begin{description}
 \item[Given:] a WG $\wg=(V,E,w)$ and a target vertex $v_{\mathrm{t}}\in V.$
 \item[Answer:] 	    for each  $v\in V$, the widest width of the paths from $v$ to $v_{\mathrm{t}}$,  that is, 
	    \begin{displaymath}
	     \max_{n\in \nat, v=v_{0},v_{1},\dotsc, v_{n}=v_{\mathrm{t}}} \min_{i\in [0,n-1]} w(v_{i},v_{i+1}),
	    \end{displaymath}
\end{description}
	We let $\wpw(\wg,v_{\mathrm{t}})$ denote a function that solves this problem, and let $\wpath(\wg,v_{\mathrm{t}})$ denote a function that assigns to each $v\in V$ a widest path to $v_{\mathrm{t}}$. Furthermore, we assume the following property of $\wpath$: if $\wpath(\wg,v_{\mathrm{t}})(v_{0})=v_{0}v_{1}\dotsc v_{k}v_{\mathrm{t}}$, then
$\wpath(\wg,v_{\mathrm{t}})(v_{i})=v_{i}v_{i+1}\dotsc v_{k}v_{\mathrm{t}}$ for each $i\in [0,k]$.
\end{mydefinition}
Efficient algorithms are known for  $\wpw(\wg,v_{\mathrm{t}})$. An example is the Dijkstra search algorithm with Fibonacci heaps~\cite{FredmanT87}; it is originally for the single-source all-destination version but its adaptation is easy. The algorithm runs in time $O(|E| + |V|\log|V|)$. It returns a widest path in addition to its width, too, computing the function $\wpath(\wg,v_{\mathrm{t}})$ with the  property required in the above.

\section{(Bounded) Value Iteration}\label{sec:VIAndBVI}
\subsection{Bellman Operator and Value Iteration}
The following construct---used for ``local propagation'' in computing the value function---is central to formal analysis of probabilistic systems and games.

\begin{mydefinition}[Bellman Operator]\label{def:bellman}
Let $\game
=(S, S_{\square}, S_{\bigcirc}, s_I,
	 \mathbf{1}, \mathbf{0}, A, \Av, \delta)
$ be a stochastic game. 
For each state $s \in S$, an available action $a \in \Av(s)$, and $f:S \to [0,1]$, we  define a function  $\mynext_a f:S \to [0,1]$ by the following. 
$$
    (\mynext_a f)(s) = 
  \begin{cases} 
   1 & \text{if } s = \mathbf{1}, \\
   0 & \text{if } s = \mathbf{0}, \\
   \sum_{s' \in S} \delta(s,a,s') \cdot f(s') & \text{if } s \neq \zero, \one.
  \end{cases}
$$
These functions are used in the following definition of the \emph{Bellman operator} $\mynext\colon
(S\to [0,1])\to (S\to [0,1])$ over $\game$:
$$
    (\mynext f)(s) = 
  \begin{cases} 
   \max_{a \in \Av(s)} (\mynext_a f)(s) & \text{if } s \in S_\square \text{ is a Maximizer state,}\\
   \min_{a \in \Av(s)} (\mynext_a f)(s) & \text{if } s \in S_\bigcirc \text{ is a Minimizer state.}
  \end{cases}
$$
\end{mydefinition}

The function space $S\to [0,1]$ inherits
 the usual order $\le$ between real numbers in the unit interval $[0,1]$, that is, $f\le g$ if  $f(s)\le g(s)$ for each $s\in S$. The Bellman operator $\mynext$ over $S\to [0,1]$ is clearly monotone; it is easily  seen to preserve $\max$ and $\min$, using the fact that the state space $S$ of an SG is finite. Therefore we obtain the following, as consequences of Kleene's fixed point theorem. 
\begin{mylemma}\label{lem:bellmanFixedPt}
 Assume the setting of Def.~\ref{def:bellman}. 
 \begin{enumerate}
  \item\label{item:bellmanGreatest} The Bellman operator $\mynext$ has the greatest fixed point (gfp) $\nu \mynext\colon S\to [0,1]$. It is obtained as the limit of the descending $\omega$-chain
 \begin{displaymath}
  \top \;\ge\; \mynext\top \;\ge\; \mynext^{2}\top \;\ge\; \cdots,
 \end{displaymath}
 where $\top$ is the greatest element of $S\to [0,1]$ (i.e., $\top(s)=1$ for each $s\in S$). In other words, we have 
 \begin{math}
  (\nu\mynext)(s)=\inf_{i\in \nat}\bigl((\mynext^{i}\top)(s) \bigr)
 \end{math} for each $s\in S$. 
  \item\label{item:bellmanLeast} Symmetrically, $\mynext$ has the least fixed point (lfp) $\mu \mynext\colon S\to [0,1]$, obtained as the limit of the ascending chain 
 \begin{equation}\label{eq:kleeneChainForValueFunc}
  \bot \;\le\; \mynext\bot \;\le\; \mynext^{2}\bot \;\le\; \cdots,
 \end{equation} where $\bot(s)=0$ for each $s\in S$. That is, we have
 \begin{math}
  (\mu\mynext)(s)=\sup_{i\in \nat}\bigl((\mynext^{i}\bot)(s) \bigr)
 \end{math} for each $s\in S$. \qed
  \end{enumerate}
\end{mylemma}

The following characterization is fundamental. See, e.g.,~\cite{ChatterjeeH08}.
\begin{mytheorem}\label{thm:valueFuncAsFixedPt}
 Let  $\game$ be a stochastic game. 
 The value function $V=V(\game)$  (Def.~\ref{def:valueFunc}) coincides with the least fixed point $\mu \mynext$. \qed
\end{mytheorem}
The fact that $V(\game)$ is the least fixed point of $\mynext$ implies the following:
a strategy $\tau$ of Maximizer is optimal if and only if $\bigl(\mynext_{\tau(s)} \bigl(V(\game)\bigr)\bigr)(s) =V(\game)(s)$ holds for each $s \in S_\square$;  similarly for Minimizer.
We say $a \in \Av(s)$ is \emph{optimal}  at $s$ if $\mynext_a V(\game)(s) =V(\game)(s)$ holds; otherwise $a$ is  \emph{suboptimal}.

Lem.~\ref{lem:bellmanFixedPt}.\ref{item:bellmanLeast} \& Thm.~\ref{thm:valueFuncAsFixedPt} suggest iterative \emph{under-}approximation of $V(\game)$
 by $\bot\le \mynext\bot\le \mynext^{2}\bot\le \cdots$. This is the principle of \emph{value iteration} (VI); see Algorithm~\ref{algo:VI}.

	\begin{algorithm}[t]
	\caption{Value iteration (VI) for a stochastic game   $\game= (S,S_{\maximizer}, S_{\minimizer},s_I,\mathbf{1},\mathbf{0},A,\Av,\delta)$ and a stopping threshold $\Delta > 0$}\label{algo:VI}

		\DontPrintSemicolon
		\SetKwProg{Pn}{procedure}{}{}
		\Pn{\textnormal{VI(}$ \game, \Delta$\textnormal{)}}{
			$L_{0} \leftarrow \bot$ 
                        \tcp*[r]{Initialize lower bound}
			\While(\tcp*[f]{Typical stopping criterion}){$L_{i}(s_{I})-L_{i-1}(s_{I})<\Delta$}{
                                $i$++ \\
				$L_{i} \leftarrow \mynext L_{i-1}$ \tcp*[r]{Bellman update}
			}
			\Return{$L_{i}(s_I)$}
		}
	\end{algorithm}

\begin{myexample} \label{ex_L}
The values $L_i$ computed by  Algorithm~\ref{algo:VI}, for the SG in Fig.~\ref{fig:ex_sg}, are shown in the following table. 
The values at $\zero$ and $\one$ are omitted.  
%
\begin{displaymath}
\scalebox{.8}{     \begin{tabular}{C{1cm}||C{1cm}|C{1cm}|C{1cm}|C{1cm}|C{1cm}|C{1.2cm}|C{1cm}||C{1cm}}
 $s$ & $L_{0}$ & $L_1$ & $L_2$ & $L_3$ & $L_4$ & $L_5$ & ... & $V(\game)$ \\ \hline\hline
 $s_I$             & 0           & 0           & 0.54        & 0.83        & 0.872       & 0.8888      &           & 0.9        \\ \cline{1-7} \cline{9-9} 
 $s_1$             & 0           & 0           & 0.8         & 0.8         & 0.8         & 0.8         & ...             & 0.8        \\ \cline{1-7} \cline{9-9} 
 $s_2$             & 0           & 0.9         & 0.9         & 0.9         & 0.9         & 0.9         &              & 0.9        
 \end{tabular}
}\end{displaymath}
 $L_i(s_I)$ converges to, but is never equal to, $V(\game)(s_I)$. 
The converges rate  can be arbitrarily slow: for any $\varepsilon \in (0,1)$ and $k \in \nat$ there is an SG $\game$ and a state $s$ such that $V(\game)(s) - L_k(s) >\varepsilon$. 
One sees this by modifying Fig.~\ref{fig:ex_sg}
with $\delta(s_I, \gamma, s_2) = \varepsilon'$ and $\delta(s_I, \gamma, s_I) = 1- \varepsilon'$, where $\varepsilon' > 0$ is an arbitrary small positive constant.
\end{myexample}

 There is no known stopping criterion for VI (Algorithm~\ref{algo:VI}) with a precision guarantee, besides the one in~\cite{ChatterjeeH08} that is too pessimistic to be practical. The one shown in Line~3 (``little progress'') is a  commonly used heuristic, but it is known to lead to arbitrarily wrong results~\cite{HADDAD2018111}.


\subsection{Bounded Value Iteration }\label{subsec:bvi}

When we turn back to  Lem.~\ref{lem:bellmanFixedPt},  Lem.~\ref{lem:bellmanFixedPt}.\ref{item:bellmanGreatest} suggests another iterative approximation, namely \emph{over-}approximation of the value function $V$ by $\top\ge \mynext\top\ge \mynext^{2}\top\ge \cdots$. The chain converges to the gfp $\nu \mynext$ that is necessarily above the lfp $\mu\mynext$. This is the principle that underlies \emph{bounded value iteration} (BVI); see Algorithm~\ref{algo:BVI} for its naive prototype. BVI has been actively studied in the literature~\cite{McMahanLG05,Brazdil2014verification,HADDAD2018111,kelmendi2018value,AshokKW19,EisentrautKA19preprint,AshokKW19preprint}, sometimes under different names (such as \emph{bounded real time dynamic programming}~\cite{McMahanLG05,Brazdil2014verification} or \emph{interval iteration}~\cite{HADDAD2018111}).

	\begin{algorithm}[t]
	\caption{Bounded value iteration (BVI) for a stochastic game   $\game
= (S,S_{\maximizer}, S_{\minimizer},s_I,\mathbf{1},\mathbf{0},A,\Av,\delta)$ and a stopping threshold $\varepsilon > 0$---a naive prototype that suffers from end components}\label{algo:BVI}

		\DontPrintSemicolon
		\SetKwProg{Pn}{procedure}{}{}
		\Pn{\textnormal{VI(}$\game, \varepsilon$\textnormal{)}}{
			$L_{0} \leftarrow \bot$, $U_{0} \leftarrow \top$ \tcp*[r]{Initialize lower and upper bound}
			\While(\tcp*[f]{Check the gap at the initial state}){$U_{i}(s_I) - L_{i}(s_I) > \varepsilon$}{
                                $i$++ \\
				$L_{i} \leftarrow \mynext L_{i-1}, \; U_{i} \leftarrow \mynext U_{i-1}$ \tcp*[r]{Bellman update}
			}
			\Return{$L_{i}(s_I)$}
		}
	\end{algorithm}

BVI comes with a precision guarantee: since $V(\game)$ lies between $L_{i}$ and $U_{i}$ (whose gap is at most $\varepsilon$), the approximation $L_{i}$ is at most $\varepsilon$ apart from $V(\game)$. 

The catch, however, is that $\mu\mynext$ and $\nu\mynext$ may not coincide, and therefore the overapproximation might not converge to  the desired $\mu \mynext$. This means Algorithm~\ref{algo:BVI} might not terminate. This is the main technical challenge addressed in the previous works on BVI, including~\cite{Brazdil2014verification,kelmendi2018value}. 

In those works, the source of the failure of convergence is identified to be \emph{end components}. See the (very simple) Markov chain in Fig.~\ref{fig:AMarkovchainforwhichthenaiveBVIfailstoconverge}, where the reachability probability from $s_{I}$ to $\one$ is clearly $0$. However, due to the loop between $s_{I}$ and $s$, the values $U_{i}(s_{I})$ and $U_{i}(s)$---these get updated to the average of $U_{i-1}$ at successors---are easily seen to remain $1$. Roughly speaking, end components generalize such loops defined in MDPs and SGs (the definitions are graph-theoretic, in terms of strongly connected components). End components cause non-convergence of naive BVI, essentially for the reason  we just described. 

The solutions previously proposed to this challenge have been to ``get rid of end components.'' For MDPs (1.5 players), the \emph{collapsing} technique  detects end components and collapses each of them into a single state~\cite{Brazdil2014verification,HADDAD2018111}. After doing so, the Bellman operator $\mynext$ has a unique fixed point (therefore $\mu\mynext=\nu\mynext$), assuring convergence of BVI (Algorithm~\ref{algo:BVI}). In the case of SGs (2.5 players), end components cannot simply be collapsed into single states---they must be handled carefully, taking the ``best exits'' into account. This is the key idea of the \emph{deflating} technique proposed for SGs in~\cite{kelmendi2018value}.

\section{Our Algorithm: Bounded Value Iteration with Upper Bounds Given by Widest Paths}\label{sec:ourAlgorithm}
In our algorithm,  like in other BVI algorithm, we iteratively construct upper and lower bounds $U_{i}, L_{i}$ of the value function $V(\game)$ at the same time. See~(\ref{eq:LandUIntro}).
 In updating $U_{i}$, however, we go beyond the \emph{local} propagation by the Bellman update and conduct \emph{global} propagation, too. This frees us from the curse of end components. 
The outline of our algorithm is as follows.
\begin{itemize}
 \item The lower bound $L_{i}$ is given by $L_{i}=\mynext^{i}\bot$, following Lem.~\ref{lem:bellmanFixedPt}.\ref{item:bellmanLeast} and Thm.~\ref{thm:valueFuncAsFixedPt}. This is the same as the other VI algorithms. 
 \item The upper bounds $U_{i}$ is constructed in the following three steps, using a \emph{global} propagation that takes advantage of fast widest path algorithms.
       \begin{itemize}
	\item \textbf{(Player reduction)} Firstly, we turn the SG $\game$ into an MDP $\mdp_{i}$ by fixing Minimizer's strategy to a specific one $\sigma_{i}$.  

Any choice of $\sigma_{i}$ would do for the sake of \emph{soundness} (that is, $V(\game)\le U_{i}$). However, for \emph{convergence} (that is, $U_{i}\to V(\game)$ as $i\to\infty$),  it is important to have $\sigma_{0}, \sigma_{1},\dotsc$ eventually converge to Minimizer's optimal strategy $\sigma_{\minimizer}$. Therefore we let $L_{i}$---the current lower estimate of $V(\game)$---induce $\sigma_{i}$. Recall that $L_{i}$ converges to $V(\game)$ (Lem.~\ref{lem:bellmanFixedPt}.\ref{item:bellmanLeast}, Thm.~\ref{thm:valueFuncAsFixedPt}). 
	\item  \textbf{(Preprocessing by local propagation)} Secondly, we turn the MDP $\mdp_{i}$ into a weighted graph (WG) $W_{i}$. 

 The construction here is \emph{local} propagation of the previous upper bound $U_{i-1}$, from each state $s$ to its predecessors in $\mdp_{i}$. This is much like an application of the Bellman operator $\mynext$.
	\item \textbf{(Global propagation by widest paths)} Finally, we solve the widest path problem in the WG $W_{i}$, from each state $s$ to the target state $\one$.  The maximum path width from $s$ to $\one$ is used as the value of the upper bound  $U_{i}(s)$. 

This way, we conduct \emph{global} propagation of upper bounds, for which end components pose no threats. Our global propagation is still computationally feasible, thanks to the preprocessing in the previous step that turns a problem on an MDP into one on a WG (modulo some sound approximation).
       \end{itemize}
\end{itemize}
The use of \emph{global} propagation for upper bounds is a distinguishing feature of our algorithm.
This is unlike other BVI algorithms (such as~\cite{Brazdil2014verification,kelmendi2018value}) where upper-bound propagation is only local and stepwise. The latter gets trapped when it encounters an EC---therefore some trick such as collapsing~\cite{Brazdil2014verification} and deflating~\cite{kelmendi2018value} is needed---while our global propagation looks directly at the target state $\one$.

The above outline is presented as pseudocode in Algorithm~\ref{algo:ourAlgorithm}. We describe the three steps in the rest of the section. In particular, we exhibit
the definitions of $\playerReduction $ and $\localPropagation $ ($\wpw $ has been defined and discussed in Def.~\ref{def:wpp}), providing some of their properties towards the correctness proof of the algorithm (\S{}\ref{subsec:soundnessAndConvergence}).

\begin{algorithm}[tbp]
	\caption{Our BVI algorithm via widest paths. Here  $\game= (S,S_{\maximizer}, S_{\minimizer},s_I,\mathbf{1},\mathbf{0},A,\Av,\delta)$ is an SG; $\varepsilon > 0$ is a stopping threshold.}\label{algo:ourAlgorithm}
	\DontPrintSemicolon
	\SetKwProg{Pn}{procedure}{}{}
	\Pn{\textnormal{BVI\_WP(}$
                                   \game, \varepsilon
                                   $\textnormal{)}}{
		$L_0 \leftarrow \bot, \; U_0 \leftarrow \top, \; i \leftarrow 0$
		\BlankLine
		\While{$U_i(s_I) - L_i(s_I) > \varepsilon$\label{algline:loopstart}}{
			$i$++ \\
                        $L_i \leftarrow \mynext L_{i-1}$ 	
\label{algline:VI}
\tcp*{value iteration for lower bounds}
                        $\mdp_{i} \leftarrow \playerReduction (\game, L_{i})$\label{algline:playerReduction} 
\tcp*{player reduction}
			$\wg_{i} \leftarrow \localPropagation (\mdp_{i}, U_{i-1})$\label{algline:localPropagation} 
\tcp*{local propagation}
			$U_i \leftarrow \min\{U_{i-1}, \wpw  (\wg_{i})\}$\label{algline:WPP}
\tcp*{widest path computation}
		}
		\BlankLine
		\Return{$U_i(s_I)$}
	}
\end{algorithm}

\subsection{Player Reduction: from SGs to MDPs}
The following general definition is not directly used in Algorithm~\ref{algo:ourAlgorithm}. It is used in our theoretical development below, towards the algorithm's correctness.
\begin{mydefinition}[the MDP $\mdp(\game, \Av')$]\label{def:MDPbyActionRestriction}
 Let $\game$ be the game in Algorithm~\ref{algo:ourAlgorithm}, and $\Av'\colon S\to 2^{A}$ be  such that $\emptyset\neq\Av'(s)\subseteq \Av(s)$ for each $s\in S$. 

Then the MDP given by the tuple
$(S,S\setminus\{\zero\},\{\zero\}, s_{I},\mathbf{1},\mathbf{0},A,\Av',\delta)$ shall be denoted by
 $\mdp(\game, \Av')$, and we say it is induced from $\game$ by restricting $\Av$ to $\Av'$. 
\end{mydefinition}
 The above construction consists of 1) restricting actions (from $\Av$ to $\Av'$), and 2) turning Minimizer's states into Maximizer's.

The following class of action restrictions will be heavily used.
\begin{mydefinition}[Minimizer restriction]\label{def:minimizerRestriction}
 Let $\game$ be
as in Algorithm~\ref{algo:ourAlgorithm}. A \emph{Minimizer restriction} of $\Av$
is a function $\Av'\colon S\to 2^{A}$ such that 
1) $\emptyset\neq \Av'(s)\subseteq \Av(s)$  for each $s\in S$, and
2) $\Av' (s) = \Av(s)$ for each state $s\in S_{\maximizer}$ of Maximizer's.
\end{mydefinition}

In Algorithm~\ref{algo:ourAlgorithm}, we will be using the MDP induced by the following specific Minimizer restriction induced by a function $f$. 


\begin{mydefinition}[the MDP $\playerReduction (\game, f)$]\label{def:playerReduction}
 Let $\game$ be the game in Algorithm~\ref{algo:ourAlgorithm}, and $f\colon S\to [0,1]$ be a function. The MDP $\playerReduction (\game, f)$
is defined to be $\mdp(\game, \Av_{f})$ (Def.~\ref{def:MDPbyActionRestriction}),
where the function $\Av_{f}\colon S\to 2^{A}$ is defined as follows.
\begin{align}
\Av_{f} (s) &= \Av(s)
&&\text{for $s\in S_{\maximizer}$,}
\nonumber
 \\
 \Av_{f} (s) &= \{a \in \Av(s) \mid \forall b \in \Av(s).\, (\mynext_{a} f)(s) \leq (\mynext_{b} f)(s)\}
&&\text{for $s\in S_{\minimizer}$.}
\label{eq:Avf}
\end{align}
The function $\Av_{f}$ is a Minimizer restriction in $\game$  (Def.~\ref{def:minimizerRestriction}).
\end{mydefinition}
\noindent 
The intuition of~(\ref{eq:Avf}) is that $a=\argmin_{b\in \Av(s)}(\mynext_{b}f)(s)$. In the use of this construction in  Algorithm~\ref{algo:ourAlgorithm}, the function $f$ will be our ``best guess'' $L_{i}$ of the value function $V(\game)$. In this situation, $\argmin_{b\in \Av(s)}(\mynext_{b}f)(s)$ is the best action for Minimizer based on the guess $f=L_{i}$. 

\begin{mydefinition}[the MDP $\mdp_{i}$, and $\Av_{i}$]\label{def:MDPI}
 In Algorithm~\ref{algo:ourAlgorithm}, the MDP $\mdp_{i}$ is given by $\playerReduction(\game, L_{i})=\mdp(\game,\Av_{L_{i}})$. We write $\Av_{i}$  for available actions in $\mdp_{i}$, that is, 
 $\mdp_{i}=(S, \one,\zero, A,\Av_{i}, \delta)$.
\end{mydefinition}

In the case of Algorithm~\ref{algo:ourAlgorithm}, the MDPs $\mdp_{0},\mdp_{1},\dotsc$ do not only ``converge'' to $\game$, but also ``reach $\game$ in finitely many steps,'' in the following sense. 
The proof relies crucially on the fact that the set $\Av(s)$ of available actions is 
finite---there is uniform $\varepsilon >0$ such that every suboptimal action is suboptimal by a gap at least $\varepsilon$. 
\begin{mylemma}\label{lem:MIReachesG}
 In Algorithm~\ref{algo:ourAlgorithm}, there exists $i_{\mathrm{M}}\in \nat$ such that, for each $i\ge i_{\mathrm{M}}$, we have $V(\game)=V(\mdp_{i})$. 
\end{mylemma}

\begin{proof}
 The proof outline is as follows. By Lem.~\ref{lem:bellmanFixedPt}.\ref{item:bellmanLeast} and Thm.~\ref{thm:valueFuncAsFixedPt}, $L_{i}$ converges to $V(\game)$. It is important that the set $\Av(s)$ of available actions is discrete in~(\ref{eq:Avf}): in case $\Av_{i}$ (Def.~\ref{def:MDPI}) consists of suboptimal actions of Minimizer's, it has to be because $L_{i}$ is at least $\varepsilon$ apart from $V(\game)$ for some fixed $\varepsilon$. This does not happen when $i$ is large enough. 

More precisely, let $\varepsilon >0$ be a positive number such that, for each $s\in S_{\minimizer}$ and each $a\in \Av(s)$, 
\begin{equation}\label{eq:convSuboptim}
  \begin{aligned}
&  (\mynext_a (V(\game)))(s) > \min_{b \in \Av(s)}(\mynext_b (V(\game)))(s)
\\&\Longrightarrow\quad
	(\mynext_a (V(\game)))(s) \geq \min_{b \in \Av(s)}(\mynext_b (V(\game)))(s) +\varepsilon.
 \end{aligned}
\end{equation}
Such $\varepsilon$ exists since there are only finitely many $s$ and $a$. 

Now let $i_{\mathrm{M}}$ be such that 
\begin{displaymath}
 V(\game)(s)-L_{i_{\mathrm{M}}}(s) < \varepsilon\quad\text{for each $s\in S$;}
\end{displaymath}
such $i_{\mathrm{M}}$ exists since $L_{i}$ converges to $V(\game)$. Since $L_{i}$ is increasing, we have 
\begin{math}
  V(\game)(s)-L_{i}(s) \le \varepsilon
\end{math}
for each $s$ and $i\ge i_{\mathrm{M}}$.

Let $s\in S_{\minimizer}$ be an arbitrary Minimizer state,  $a\in \Av(s)$ be any suboptimal action (with respect to $V(\game)$), and $a^{*}$ be an optimal one. It suffices to show that $a\not\in\Av_{i}(s)$ for each $i\ge i_{\mathrm{M}}$.

By~(\ref{eq:convSuboptim}) we have 
\begin{math}
 	(\mynext_a (V(\game)))(s) \geq (\mynext_{a^{*}} (V(\game)))(s) +\varepsilon
\end{math}. Moreover, 
	\begin{align*}
	(\mynext_a (V(\game)))(s) 
	&=    \sum_{s' \in S} \delta(s,a,s')\cdot V(\game)(s') \\
	&< \biggl( \sum_{s' \in S} \delta(s,a,s')\cdot L_i(s') \biggr) + \varepsilon = (\mynext_a L_i)(s) + \varepsilon.
	\end{align*}
The last two observations are used in the following.
\begin{align*}
 	(\mynext_a L_i)(s) 
	\;>\; (\mynext_a (V(\game)))(s) -\varepsilon
	\;\geq\; 
        (\mynext_{a^{*}}( V(\game)))(s) 
	\;\geq\; (\mynext_{a^{*}} L_i)(s),
\end{align*}
where the last inequality is because $L_{i}$ is a lower bound of $V(\game)$. 
Therefore $a\not\in\Av_{i}(s)$ by~(\ref{eq:Avf}). 
\qed
\end{proof}

\subsection{Local Propagation: from MDPs to WGs}
Here is a  technical observation that motivates the function $\localPropagation $.

\begin{mylemma}\label{lem:keyLemmaForBVIWPP}
  Let $\game$ be the game in Algorithm~\ref{algo:ourAlgorithm}, and $\Av'\colon S\to 2^{A}$ be a Minimizer restriction (Def.~\ref{def:minimizerRestriction}).
\begin{enumerate}
 \item For each state
  $s\in S$, we have 
 \begin{math}
 V(\game)(s)
 \le
 \textstyle\max_{a\in \Av'(s)}\bigl(\,\mynext_{a}\bigl(V(\game)\bigr)\,\bigr)(s).
 \end{math}
 \item 
For each $k\in \nat$, we have
 \begin{equation}\label{eq:almostWPP}
  V(\game)(s_{0})
  \;\le\;
  \max_{s_{0}\xrightarrow{a_{0}}s_{1}\xrightarrow{a_{1}}\,\cdots\,\xrightarrow{a_{k}}
 \text{ in $\Av'$}
} 
 \bigl(\,\mynext_{a_{k}}\bigl(V(\game)\bigr)\,\bigr)(s_{k}),
 \end{equation}
 where the maximum is taken over $a_{0}, s_{1},a_{1},\dotsc, s_{k}, a_{k}$ such that
\begin{math}
 a_{0}\in \Av'(s_{0}), s_{1}\in \post(s_{0},a_{0}), a_{1}\in \Av'(s_{1}), \dotsc,
  s_{k}\in \post(s_{k-1},a_{k-1}),
  a_{k}  \in \Av'(s_{k})
\end{math}.
\end{enumerate}
\end{mylemma}
\begin{proof}
 For the item 1, recall that $V(\game)$ is the least fixed point of the Bellman operator (Thm.~\ref{thm:valueFuncAsFixedPt}). For each Minimizer state $s\in S_{\minimizer}$, we have
\begin{align*}
   V(\game)(s) = \min_{a\in \Av(s)}\bigl(\mynext_{a}\bigl(V(\game)\bigr)\bigr)(s)
  \le \min_{a\in \Av'(s)}\bigl(\mynext_{a}\bigl(V(\game)\bigr)\bigr)(s)
  \le \max_{a\in \Av'(s)}\bigl(\mynext_{a}\bigl(V(\game)\bigr)\bigr)(s).
\end{align*}
 For each Maximizer state $s\in S_{\maximizer}$, we have
\begin{align*}
   V(\game)(s) = \max_{a\in \Av(s)}\bigl(\mynext_{a}\bigl(V(\game)\bigr)\bigr)(s)
= \max_{a\in \Av'(s)}\bigl(\mynext_{a}\bigl(V(\game)\bigr)\bigr)(s).
\end{align*}
The latter equality is because $\Av'$ does not restrict Maximizer's actions. This proves the item 1.

The item 2 is proved by induction as follows, using the item 1 in its course.
\begin{align*}
  & V(\game)(s_{0})
\\
&\le 
 \max_{a_{0}\in \Av'(s_{0})}\bigl(\,\mynext_{a_{0}}(V(\game))\,\bigr)(s_{0})
  \qquad\text{by the item 1.}
\numberthis\label{eq:iIsZeroFanLemma}
\\
&=
\max_{a_{0}\in \Av'(s_{0})}
 \sum_{s_{1}\in \post(s_{0},a_{0})}\delta(s_{0},a_{0},s_{1})\cdot V(\game)(s_{1})
\\
&\le
\max_{a_{0}\in \Av'(s_{0})}
 \sum_{s_{1}\in \post(s_{0},a_{0})}\delta(s_{0},a_{0},s_{1})\cdot 
\Bigl(\,  \max_{s_{1}\stackrel{a_{1}}{\to}\,\cdots\,\stackrel{a_{k}}{\to}
 \text{\,\, in $\Av'$}} 
 \bigl(\,\mynext_{a_{k}}\bigl(V(\game)\bigr)\,\bigr)(s_{k})
\,\Bigr)
\\
&\qquad\qquad  \qquad\text{by the induction hypothesis (for $k-1$)}
\\
&\le
\max_{a_{0}\in \Av'(s_{0})}
\max_{s_{1}\in \post(s_{0},a_{0})}
\Bigl(\,  \max_{s_{1}\stackrel{a_{1}}{\to}\,\cdots\,\stackrel{a_{k}}{\to}
 \text{\,\, in $\Av'$}} 
 \bigl(\,\mynext_{a_{k}}\bigl(V(\game)\bigr)\,\bigr)(s_{k})
\,\Bigr)
\numberthis\label{eq:fromAvToMax}
\\
&=
  \max_{s_{0}\xrightarrow{a_{0}}s_{1}\xrightarrow{a_{1}}\,\cdots\,\xrightarrow{a_{k}}
 \text{ in $\Av'$}
} 
 \bigl(\,\mynext_{a_{k}}\bigl(V(\game)\bigr)\,\bigr)(s_{k}).
\end{align*}
The inequality in~(\ref{eq:fromAvToMax}) holds since an average over $s_{1}$ on the left-hand side is replaced by the corresponding maximum on the right-hand side. Note that the value 
\begin{math}
  \max_{s_{1}\stackrel{a_{1}}{\to}\,\cdots\,\stackrel{a_{k}}{\to}
 \text{\,\, in $\Av'$}} 
 \min_{i\in[1,k]}
 \bigl(\,\mynext_{a_{i}}\bigl(V(\game)\bigr)\,\bigr)(s_{i})
\end{math} that
occurs on both sides is determined once $s_{1}$ is determined. 
This concludes the proof. \qed
\end{proof}
Lem.~\ref{lem:keyLemmaForBVIWPP}.2, although  not itself used in the following technical development, suggests the idea of global propagation for upper bounds.
Note that a bound is given in~(\ref{eq:almostWPP}) for each $k$; it is possible that a bound for some $k>1$ is tighter than that for $k=1$, motivating us to take a ``look-ahead'' further than one step. 

However, the bound in~(\ref{eq:almostWPP}) is not particularly tuned for tractability: computation of the maximum involves words whose number is exponential in $k$, and moreover, we want to do so for many $k$'s. 

In the end, our main technical contribution is that a similar ``look-ahead'' can be done by solving the widest path problem in the following weighted graph. The soundness of this method is not so easy as for Lem.~\ref{lem:keyLemmaForBVIWPP}.2---see~\S{}\ref{subsec:soundnessAndConvergence}. 


\begin{mydefinition}[the WG $\localPropagation (\mdp,f)$]\label{def:localPropag}
Let 
$\mdp=(S, \one,\zero, A,\Av', \delta)$ be an MDP, 
and $f\colon S\to [0,1] $.
The WG $\localPropagation (\mdp,f)$ is  the following triple $(S,E,w)$.
\begin{itemize}
 \item Its set of vertices is $S$. 
 \item 
We have $(s,s')\in E$
if and only if, for some $a\in \Av'(s)$,  we have $s'\in \post(s,a)$ (i.e., $\delta(s,a,s')>0$). 
 \item The weight function $w\colon E\to [0,1]$ is given by 
\begin{equation}\label{eq:defOfWeightInLocalPropagation}
 		w(s,s') = \max \bigl\{\,\mynext_a f(s) \;\big|\; a \in \Av'(s), 
                   s'\in\post(s,a)\,\bigr\}.
\end{equation}
\end{itemize}
\end{mydefinition}
\noindent In~(\ref{eq:defOfWeightInLocalPropagation}), the function $f$---that is,  the previous upper bound $U_{i-1}$ in Algorithm~\ref{algo:ourAlgorithm}---is propagated one step by the application of $\mynext_{a}$. This way of encoding these propagated values as weights in a WG seems pretty rough. For example, in case both $s'$ and $s''$ are in $\post(s,a)$ for each $a\in \Av'(s)$, we have $w(s,s')=w(s,s'')$, no matter what the transition probabilities from $s$ to $s',s''$ are. 
The return for this paid price (namely the information lost in the rough encoding) is that the resulting data structure (WG) allows fast \emph{global} analysis via the widest path problem. Our experiment results in~\S{}\ref{sec:experiment} demonstrate that this rough yet global approximation can make upper bounds quickly converge.

\subsection{Soundness and Convergence}\label{subsec:soundnessAndConvergence}
In Algorithm~\ref{algo:ourAlgorithm}, an SG $\game$ is turned into an MDP $\mdp_{i}$ and then to a WG $\wg_{i}$. Our claim is that computing a widest path in $\wg_{i}$ gives the next upper bound $U_{i}$ in the iteration. Here we prove the following correctness properties: soundness ($V(\game)\le U_{i}$) and convergence ($U_{i}\to V(\game)$ as $i\to\infty$).

We start with a technical lemma. The choice of the MDP $\mdp(\game,\Av')$ and the value function $V(\game)$ (for $\game$, not for  $\mdp(\game,\Av')$) in the statement is subtle; it turns out to be just what we need. 
 \begin{mylemma}\label{lem:increasingPath}
   Let $\game$ be as in Algorithm~\ref{algo:ourAlgorithm}, and $\Av'\colon S\to 2^{A}$ be a Minimizer restriction (Def.~\ref{def:minimizerRestriction}). Let $s_{0}\in  S_{\Diamond\one}$ be a state with a non-zero value (Def.~\ref{def:valueFunc}).  Consider
the MDP $\mdp(\game,\Av')$ (Def.~\ref{def:MDPbyActionRestriction}), for which we write simply $\mdp$.
Then
 there is a finite path $\pi=s_0 a_0 s_1 a_1 \dotsc a_{n-1} s_{n}$ in $\mdp$ that satisfies the following.
 \begin{itemize}
 \item The path $\pi$ reaches $\one$, that is, $s_{n}=\one$. 
 \item Each action is optimal in  $\mdp$ with respect to $V(\game)$, that is, 
  $
  \bigl(\mynext_{a_{i}}\bigl(V(\game)\bigr)\bigr)(s_{i})
  =
  max_{a\in \Av'(s_{i})}   \bigl(\mynext_{a}\bigl(V(\game)\bigr)\bigr)(s_{i})
  $ for each $i\in [0,n-1]$.
 \item The value function $V(\game)$ does not decrease along the path, that is, 
       \begin{math}
 	V(\game)(s_{i})\le 	V(\game)(s_{i+1})
       \end{math}
       for each $i\in [0,n-1]$.
 \end{itemize}
 \end{mylemma}
 \begin{proof}
 We construct a function $\mathsf{PATH}:  S_{\Diamond\one} \to S^+$ 
 by Algorithm~\ref{algo:increasingPath}. It is clear that $\mathsf{PATH}$ assigns a desired path to each $s_{0}\in S_{\Diamond\one}$. In particular, $V(\game)$ does not decrease along $\mathsf{PATH}(s_{0})$ since always a state with a smaller value of $V(\game)$ is prepended.

 \begin{algorithm}[tbp] 
 	\caption{A construction of $\mathsf{PATH}:  S_{\Diamond\one} \to S^+$  for Lem.~\ref{lem:increasingPath}}\label{algo:increasingPath}
 	\DontPrintSemicolon
 	\SetKwProg{Pn}{procedure}{}{}
                $S_{\mathrm{v}} \leftarrow \{\one\},\; \mathsf{PATH}(\one)\leftarrow \one$
 		 \BlankLine
 		\While{$S_{\Diamond\one} \setminus S_{\mathrm{v}}  \neq \emptyset$\label{algline:increasingPathWhile}}{
                   Choose a pair of states $(s_{\mathrm{c}},s_{\mathrm{p}})$ that satisfies the following: \qquad\qquad $s_{\mathrm{c}} \in S\setminus  S_{\mathrm{v}}$, $s_{\mathrm{p}} \in S_{\mathrm{v}}$, $V(\game)(s_{\mathrm{c}}) = \max_{s \in S\setminus S_{\mathrm{v}}} V(\game)(s)$, and 
 \qquad\qquad\qquad\quad
 for an optimal action $a$ at $s_{\mathrm{c}}$ in $\mdp$, 
 $
 s_{\mathrm{p}}\in \post(s_{\mathrm{c}},a)$
 \label{algline:chooseMaxValueState}
 \\
 		 $\mathsf{PATH}(s_{\mathrm{c}}) \leftarrow s_{\mathrm{c}} \cdot \mathsf{PATH}(s_{\mathrm{p}}),\; S_{\mathrm{v}} \leftarrow S_{\mathrm{v}} \cup \{s_{\mathrm{c}}\}$
 		}
 		\BlankLine
 		\Return{$\mathsf{PATH}$}
 \end{algorithm}

 It remains to be shown that, in Line~\ref{algline:chooseMaxValueState}, a required pair $(s_{\mathrm{c}}, s_{\mathrm{p}})$ is always found. 
 Let $S_{\mathrm{v}} \subsetneq S_{\Diamond\one}$ be a subset with $\one \in S_{\mathrm{v}}$; here $S_{\mathrm{v}}$ is a proper subset of $  S_{\Diamond\one}$ since otherwise we should be already out of the while loop (Line~\ref{algline:increasingPathWhile}). 

 Let  
$
 S_{\max} = \{s \in S \setminus  S_{\mathrm{v}} \mid V(\game)(s) = \max_{s' \in S \setminus  S_{\mathrm{v}}} V(\game)(s')\}. 
$
 Since $S_{\mathrm{v}} \subsetneq  S_{\Diamond\one}$, we have $\emptyset\neq S_{\max}\subseteq S_{\Diamond\one}$ and thus  $V(\game)(s)> 0$ for each $s\in S_{\max}$. We also have $\one\not\in S_{\max}$ since $\one\in S_{\mathrm{v}}$.

 We argue by contradiction: assume that for any $s \in S\setminus  S_{\mathrm{v}}$,  $s'\in S_{\mathrm{v}}$, 
we have $s'\not\in\post(s,a_{s})$, where $a_{s}$  is any optimal action at $s$ in $\mdp$ with respect to $V(\game)$.

 Now let $s\in S_{\max}$ be an arbitrary element. It follows that $V(\game)(s) > 0$. 
 \begin{align*}
 V(\game)(s) 
 &\leq
 \bigl(\mynext_{a_{s}}\bigl(V(\game)\bigr)\bigr)(s)
 \\
 &\!\!\!\!
\text{using Lem.~\ref{lem:keyLemmaForBVIWPP}; here $a_{s}$ is an optimal action at $s$ in $\mdp$ with respect to $V(\game)$, 
 }
 \\
 &=
 \textstyle\sum_{s' \in S\setminus  S_{\mathrm{v}}} \delta(s, a_{s},s') \cdot V(\game)(s')
 \\
 &\qquad
 \qquad\text{by the assumption that $s'\not\in\post(s,a_{s})$ for each $s'\in S_{\mathrm{v}}$} 
 \\
 &\le
 \textstyle\sum_{s' \in S\setminus  S_{\mathrm{v}}} \delta(s, a_{s},s') \cdot V(\game)(s)
 \numberthis\label{eq:ineqThatMustBeEq}
 \\
 &\qquad
 \qquad\text{since $s\in S_{\max}$ and hence $V(\game)(s')\le V(\game)(s)$} 
 \\
 &=
 V(\game)(s)
 \qquad\text{since 
 $\textstyle\sum_{s' \in S\setminus  S_{\mathrm{v}}} \delta(s_{\mathrm{c}}, a,s') =1$.
 }
 \end{align*}
 Therefore both inequalities in the above must be equalities. In particular, for the second inequality (in~(\ref{eq:ineqThatMustBeEq})) to be an equality, we must have the weight for each suboptimal $s'$ to be $0$. That is, $\delta(s,a_{s},s')=0$ for each $s'\in (S\setminus S_{\mathrm{v}})\setminus S_{\max}$.


The above holds for arbitrary $s\in S_{\max}$. Therefore, for any strategy that is optimal in $\mdp$ with respect to $V(\game)$, once a play is in $S_{\max}$, it never comes out of $S_{\max}$, hence the play never reaches $\one$.  Moreover, an optimal strategy in $\mdp$ with respect to $V(\game)$ is at least as good as an optimal strategy for Maximizer in $\game$ (with respect to $V(\game)$), that is, the latter reaches $\one$ no more often than the former. This follows from Lem.~\ref{lem:keyLemmaForBVIWPP}. Altogether, we conclude that a Maximizer optimal strategy in $\game$ does not lead any $s\in S_{\max}$ to $\one$, i.e.,  $V(\mdp)(s)=0$ for each $s\in S_{\max}$. 
Now we come to a contradiction.
\qed
 \end{proof}

In the following lemma, we use the value function $V(\game)$ in the position of $f$ in Def.~\ref{def:localPropag}. This cannot be done in actual execution of Algorithm~\ref{algo:increasingPath}: unlike $U_{i-1}$ 
in Algorithm~\ref{algo:ourAlgorithm}, the value function $V(\game)$ is not known to us.  Nevertheless, the lemma is an important theoretical vehicle towards  soundness of Algorithm~\ref{algo:ourAlgorithm}. 
\begin{mylemma}\label{lem:stepwiseSoundness}
 Let $\game$ be the game in Algorithm~\ref{algo:ourAlgorithm}, and $\Av'\colon S\to 2^{A}$ be a Minimizer restriction (Def.~\ref{def:minimizerRestriction}). 
Let
$\mdp=\mdp(\game, \Av')$, and 
   $\wg=\localPropagation \bigl(\mdp, V(\game)\bigr)$.
Then, for each state $s\in S$, we have
$\wpw(\wg)(s,\one)\ge V(\game)(s)$.
\end{mylemma}
\begin{proof}
In what follows, we let the WG $\wg=\localPropagation
 \bigl(\mdp, V(\game)\bigr)
$ be denoted by $\wg=(S, E, w)$. 
%
Let $\pi=s_0 a_0 s_1 a_1 \dotsc a_{n-1} s_{n}$ be a path of the MDP $\mdp$ such that $s_{n}=\one$, each action is optimal in $\mdp$ with respect to $V(\game)$, and
       \begin{math}
	V(\game)(s_{i})\le 	V(\game)(s_{i+1})
       \end{math}
       for each $i\in [0,n-1]$. Existence of such a path $\pi$ is shown by   Lem.~\ref{lem:increasingPath}.
Let  $\pi'=s_0 s_1 \dotsc s_{n-1} \one$ be the path in the WG $\wg$ induced by $\pi$---we simply omit actions. 

  The path $\pi'$ satisfies the following, for each $i\in [0,n-1]$. 
  \begin{align*}
   w(s_{i},s_{i+1}) 
  &=
  \max \bigl\{\,\bigl(\mynext_a \bigl(V(\game)\bigr)\bigr)(s_{i}) \;\big|\; a \in \Av'(s_{i}), 
                   s_{i+1}\in\post(s_{i},a)\,\bigr\}
\quad\text{by Def.~\ref{def:localPropag}}
  \\
  &=
\bigl(\mynext_{a_{i}} \bigl(V(\game)\bigr)\bigr)(s_{i})
    \quad\text{since $a_{i}$ is optimal wrt.\ $V(\game)$;}
  \\
  &
  \quad\text{note that $a_{i} \in \Av'(s_{i}), 
                   s_{i+1}\in\post(s_{i},a_{i})$ hold since $\pi$ is a path in $\mdp$}
  \\
  &=
 \textstyle\max_{a\in \Av'(s)}\bigl(\,\mynext_{a}\bigl(V(\game)\bigr)\,\bigr)(s_{i})
      \quad\text{since $a_{i}$ is optimal wrt.\ $V(\game)$} 
 \\
  &\ge
 V(\game)(s_{i}) \qquad\text{by Lem.~\ref{lem:keyLemmaForBVIWPP}.}
  \end{align*}
  This observation, combined with $V(\game)(s_{0})\le V(\game)(s_{1})\le\cdots\le V(\game)(s_{n})$ (by the definition of $\pi$), implies that 
the width of the path $\pi'$ is at least $V(\game)(s_{0})$. The widest path width is no smaller than that. \qed
\end{proof}


\begin{mytheorem}[soundness]\label{thm:soundness}
 In Algorithm~\ref{algo:ourAlgorithm}, $V(\game)\le U_{i}$ holds for each $i\in \nat$. 
\end{mytheorem}
\begin{proof}
  We let the function 
 \begin{align*}
 &\min\bigl\{\,U,\, \wpw\bigl(\localPropagation\bigl(\,\mdp(\game, \Av'), U\,\bigr)\bigr)(\place, \one)\,\bigr\}\quad\colon \quad S\longrightarrow [0,1]
\\
 &\text{denoted by}\qquad
 T(\Av',U)\;\colon \quad S\longrightarrow [0,1],
 \end{align*}
clarifying its dependence on
 $\Av'$ and 
 $U\colon S\to [0,1]$. 
 Clearly, for each $i\in\nat$, we have 
$U_{i}\;=\;T(\Av_{L_{i}},U_{i-1})$.


\medskip
The rest of the proof is by induction. It is trivial if $i=0$ ($U_{0}=\top$). 
\begin{align*}
 U_{i+1} &= T(\Av_{L_{i}},U_{i})
  \\
 &\ge T(\Av_{L_{i}},V(\game)) \qquad\text{by  ind.\ hyp., and $T(\Av_{L_{i}},\place)$ is monotone}
  \\
 &=
 \min\bigl\{\,V(\game),\; \wpw\bigl(\localPropagation\bigl(\,\mdp(\game, \Av_{L_{i}}), V(\game)\,\bigr)\bigr)(\place, \one)\,\bigr\}
 \\
 &= V(\game)\qquad\text{by Lem.~\ref{lem:stepwiseSoundness}.}
\tag*{\qed}
\end{align*}
\end{proof}

It is clear that $U_{i}$ decreases with respect to $i$ ($U_{0}\ge U_{1}\ge\cdots$), by the presence of $\min$ in Line~\ref{algline:WPP}. It remains to show the following.

\begin{mytheorem}[convergence]\label{thm:convergence}
  In Algorithm~\ref{algo:ourAlgorithm}, let the while loop iterate forever. Then $U_{i}\to V(\game)$ as $i\to \infty$. 
\end{mytheorem}
\begin{proof}
We give a  proof using the infinitary pigeonhole principle. The proof is nonconstructive---it is not suited for analyzing the speed of convergence, for example---but the proof becomes simpler.

 In what follows,  we let $\mynext_{\sigma}\colon (S\to[0,1])\to (S\to[0,1])$ denote the Bellman operator on an MDP $\mdp$ induced by a strategy $\sigma$, i.e., 
$
  (\mynext_{\sigma}f)(s)
  \;:=\;
(\mynext_{\sigma(s)}f)(s).
$
The MC obtained from  an MDP $\mdp$  by fixing a strategy $\sigma$ is denoted by $\mdp^{\sigma}$.

Towards the statement of the theorem, for each $i\in\nat$, we choose a (positional) strategy $\sigma_{i}$ in the MDP $\mdp_{i}$ as follows.
\begin{itemize}
 \item For each $s\in S_{\Diamond\one}$, take the widest path $\wpath(\wg_{i}, \one)(s)=s s_{1}\dotsc \one$ in $\wg_{i}$ from $s$ to $\one$ (Def.~\ref{def:wpp}). 
Such a path from $s$ to $\one$ exists---otherwise we have $U_{i}(s)=0$, hence $V(\game)(s)=0$ by Thm.~\ref{thm:soundness}. 

 Let $\sigma_{i}(s)$ be an action that justifies the first edge in the chosen widest path, that is, $a\in \Av_{i}(s)$ such that $s_{1}\in\post(s,a)$. 
 \item For each $s\in S\setminus S_{\Diamond\one}$, $\sigma_{i}(s)$ is freely chosen from $\Av_{i}(s)$. 
\end{itemize}

It is then easy to see that 
\begin{equation}\label{eq:1713}
 \wpw(\wg_{i})(s)\le (\mynext_{\sigma_{i}} U_{i-1})(s)
 \qquad\text{for each $i\in\nat$ and $s\in S_{\Diamond\one}$. }
\end{equation}
Indeed, by the definition of $\sigma_{i}$, the right-hand side is the weight of the first edge in the chosen widest path. This must be no smaller than the widest path width, that is, the width of the chosen path.

Now, since there are only finitely many strategies for the SG $\game$, the same is true for the MDPs $\mdp_{0}, \mdp_{1}, \dotsc$ that are obtained from $\game$ by restricting Minimizer's actions. Therefore, by the infinitary pigeonhole principle, there are infinitely many  $i_{0}<i_{1}<\cdots$  such that 
$ \sigma_{i_{0}}=\sigma_{i_{1}}=\cdots\quad =: \sigma^{\dagger}. $
Moreover, we can choose them so that they are all beyond $i_{\mathrm{M}}$ in Lem.~\ref{lem:MIReachesG}, in which case we have 
\begin{equation}\label{eq:mdpLowerThanGame}
 V(\mdp_{i_{m}}^{\sigma^{\dagger}})\le V(\game)
\quad\text{for each $m\in\nat$.} 
\end{equation}
Indeed, Minimizer's actions are already optimized in $\mdp_{i}$ (Lem.~\ref{lem:MIReachesG}), and thus the only freedom left for $\sigma^{\dagger}$ is to choose suboptimal actions of Maximizer's. 

In what follows, we cut down the domain of discourse from $S\to [0,1]$ to $S_{\Diamond\one}\to [0,1]$,
i.e., 
1) every function of the type $f:S \to [0,1]$ is now seen as the restriction over $S_{\Diamond\one}$, and 
2) the Bellman operator only adds up the value of the input function over $S_{\Diamond\one}$, namely 
it is now defined by 
$\hat\mynext_a f(s) = \sum_{s' \in S_{\Diamond\one}} \delta(s,a,s') \cdot f(s')$. 
The operator $\hat\mynext_\sigma$ is also defined in a similar way to $\mynext_\sigma$. 

Now 
proving convergence in $S_{\Diamond\one}\to [0,1]$ suffices for the theorem. 
Indeed, for each $i\ge i_{\mathrm{M}}$, we have $V(\mdp_{i})(s)=V(\game)(s)=0$ for each $s\in S\setminus S_{\Diamond\one}$. This implies that there is no path from $s$ to $\one$ in $\mdp_{i}$, thus neither in the WG $\wg_{i}$. Therefore $U_{i}\le\wpw(\wg_{i})=0$. 

A benefit of this domain restriction is that the Bellman operator $\hat\mynext_{\sigma}$ 
has a unique fixed point in  $S_{\Diamond\one}\to [0,1]$ 
if the set of non-sink states in $\mdp^\sigma$ is exactly $S_{\Diamond\one}$, 
i.e., $ V(\mdp^\sigma)(s) > 0 $ holds if and only if $s \in S_{\Diamond\one}$.
Furthermore, this unique fixed point is the value function $V(\mdp^{\sigma})$ restricted to $S_{\Diamond\one}\subseteq S$~\cite[Thm.~10.19]{baier2008principles}. 
Therefore $V(\mdp^{\sigma})$ is computed by the gfp Kleene iteration, too:
\begin{equation}\label{eq:uniqueFixedPointThusGfp}
\top \;\ge\;
\hat\mynext_{\sigma}\top
\;\ge\;
(\hat\mynext_{\sigma})^{2}\top
\;\ge\;
\cdots\quad\longrightarrow V(\mdp^{\sigma}) 
\quad \text{in the space $S_{\Diamond\one}\to [0,1]$.}
\end{equation} 

We show the following by induction on $m$. 
\begin{equation}\label{eq:UAndX}
U_{i_{m}} \;\le\; (\hat\mynext_{\sigma^{\dagger}})^{m}\top
\quad\text{for each $m\in\nat$.} 
\end{equation}
It is obvious for $m=0$. For the step case, 
we have the following.  
Notice that the inequality~(\ref{eq:1713}) holds in the restricted domain for $i \geq i_M$. 
\begin{align*}
U_{i_{m+1}}
&\le 
\wpw(\wg_{i_{m+1}})
\quad\text{by Line~\ref{algline:WPP} of Algorithm~\ref{algo:ourAlgorithm}} 
\\
&\le
\hat\mynext_{\sigma^{\dagger}} U_{i_{m+1}-1}
\quad\text{by~(\ref{eq:1713})}
\\
&\le
\hat\mynext_{\sigma^{\dagger}} U_{i_{m}}
\quad\text{by monotonicity of $\hat\mynext_{\sigma^{\dagger}}$, decrease of $U_{i}$ and $i_{m}<i_{m+1}$}
\\
&\le 
(\hat\mynext_{\sigma^{\dagger}})^{m+1}\top
\quad\text{by the induction hypothesis.}
\end{align*}
We have proved~(\ref{eq:UAndX}) which proves $\inf_{i}U_{i}\le \inf_m(\hat\mynext_{\sigma^{\dagger}})^{m}\top$. 

Lastly, we prove that $ V(\mdp_{i_m}^{\sigma^\dagger})(s) > 0 $ holds if and only if $ s \in S_{\Diamond\one}$ for each $m \in \nat$, 
and thus $\sigma^\dagger$ follows the characterization in~(\ref{eq:uniqueFixedPointThusGfp}).  
This proves 
\begin{equation}\label{eq:convergence}
\inf_{i}U_{i}\le V(\mdp_{i_{m}}^{\sigma^{\dagger}}) 
\quad\text{for each $m\in\nat$.} 
\end{equation}
Implication to the right is clear as Minimizer restriction is done optimally in $\mdp_{i_m}$. Conversely, if $s \in S_{\Diamond\one}$, then there is a path from $s$ to $\one$ in $\wg_{i_m}$. 
Let $\wpath(\wg_{i_m}, \one)(s) = s_0s_1\ldots s_k$, where $s_0 = s$, $k \in \nat$ and $s_k = \one$. 
Then by the property of $\wpath$ and $\sigma^\dagger$, we have $\delta(s_j, \sigma^\dagger(s_j), s_{j+1}) >0$ for each $j <k$. 
Thus, the probability that the finite path $\wpath(\wg_{i_m}, \one)(s)$ is obtained by running $\mdp_{i_m}^{\sigma^\dagger}$ starting from $s$, which is apparently at most $V(\mdp_{i_{m}}^{\sigma^{\dagger}})(s)$, is nonzero. 
Hence we have implication to the left.

Combining~(\ref{eq:mdpLowerThanGame}),~(\ref{eq:convergence}) and Thm.~\ref{thm:soundness}, we obtain the claim. \qed

\end{proof}

An example execution of Algorithm~\ref{algo:ourAlgorithm} is found in the appendix.

\section{Experiment Results}\label{sec:experiment}
\paragraph{Experiment Settings}
We compare the following four algorithms.

\begin{itemize}
 \item {\it WP} is  our BVI algorithm via  widest paths. It avoids end component (EC) computation by global propagation of upper bounds. 
 \item 
 {\it DFL} is the implementation of the main algorithm in~\cite{kelmendi2018value}. It relies on  EC computation for deflating. 
 \item 
{\it DFL\_m} is our modification of  DFL, where some unnecessary repetition of EC computation is removed.
 \item 
{\it DFL\_BRTDP} is the learning-based variant of DFL.
It
restricts  bound  update to those states which are  visited by simulations. See~\cite{kelmendi2018value} for details. 
\end{itemize}
The latter three---coming from~\cite{kelmendi2018value}---are the only existing BVI algorithms for SGs with a convergence guarantee, to the best of our knowledge. The implementation of DFL and DFL\_BRTDP is provided by the authors of~\cite{kelmendi2018value}.

The four algorithms are implemented on top of PRISM-games~\cite{KPW18} version 2.0 
and publicly available online\footnote{\url{https://github.com/kittiphonp/CAV20Impl}}. 
We used the stopping threshold $\varepsilon=10^{-6}$. 
The experiments were conducted on 
 Dell Inspiron 3421 Laptop with 4.00 GB RAM and Intel(R) Core(TM) i5-3337U 1.80 GHz processor. 

In the implementations of DFL and DFL\_BRTDP, the deflating operation is applied only once  every five iterations~\cite[\S{}B.3]{kelmendi2018value}.
Following this, our WP also solves the widest path problem (Line~\ref{algline:WPP}) only once every five iterations, while other operations are 
applied in each iteration.

For input SGs, we took four  models from the literature: {\it mdsm}~\cite{ChenFKPS13}, {\it cloud}~\cite{CalinescuKJ12}, {\it teamform}~\cite{ChenKPS11} and {\it investor}~\cite{McIverM07}. In addition, we used our model {\it manyECs}---an artificial model with many ECs---to assess the effect of ECs on performance. The model manyECs is presented in the appendix. 
Each of these five models comes with a model parameter $N$.


There is another model called {\it cdmsn} in~\cite{kelmendi2018value}. We do not discuss cdmsn  since all the algorithms (ours and those from~\cite{kelmendi2018value}) terminated within 0.001 seconds.


\paragraph{Results} 
The number $i$ of iterations  and the running time for each algorithm and each input SG is shown in Table~\ref{tab:result}. 
For DFL\_BRTDP, the ratio of states visited by the algorithm is shown in percentage; the smaller it is, the more efficient the algorithm is in reducing the state space. 
Each number for DFL\_BRTDP (a probabilistic algorithm) is the average over 5 runs. 
 

\begin{table}[tbp]
\centering
    \caption{Experimental results, comparing WP (our algorithm)  with those in~\cite{kelmendi2018value}. $N$ is a model parameter (the bigger the more complex). \#states, \#trans, \#EC show the numbers of states, transitions and ECs in the SG, respectively. itr is the number $i$ of iterations at termination; time is the execution time in seconds. For each SG, the fastest algorithm is shaded in green. The settings that did not terminate are shaded in gray; TO is time out (6 hours), OOM is out of memory, and SO is stack overflow. }
    \label{tab:result}
    \scalebox{0.94}{
\begin{tabular}{l|rrrr||rr|rr|rrr||rr}
	\multirow{2}{*}{model}    & \multirow{2}{*}{$N$} & \multirow{2}{*}{\#states} & \multirow{2}{*}{\#trans} & \multirow{2}{*}{\#EC} & \multicolumn{2}{c|}{DFL} & \multicolumn{2}{c|}{DFL\_m} & \multicolumn{3}{c||}{DFL\_BRTDP} & \multicolumn{2}{c}{WP} \\ 
	&                    &                           &                          &                       & itr    & time            & itr      & time             & itr      & visit\%    & time    & itr    & time            \\ \hline\hline
	\multirow{2}{*}{mdsm}     & 3                  & 62245                     & 151143                   & 1                     & 121    & \cellcolor{green!25}3               & 121         & 4                & 17339    & 49.3       & 15      & 120    & 5               \\ 
	& 4                  & 335211                    & 882765                   & 1                     & 125    & \cellcolor{green!25}15              & 125         & 47               & 91301    & 42.1       & 86      & 124    & 38              \\ \hline
	\multirow{3}{*}{cloud}    & 5                  & 8842                      & 60437                    & 4421                  & 7      & 7               & 7          & 1                & 167      & 6.9        & 14      & 7      & \cellcolor{green!25}\textless{}1    \\ 
	& 6                  & 34954                     & 274965                   & 17477                 & 11     & 177             & 11         & 5                & 41       & 0.6        & 3       & 11     & \cellcolor{green!25}1               \\ 
	& 7                  & 139402                    & 1237525                  & 69701                 & 11     & 19721           & 11         & 62               & 41       & 0.2        & \cellcolor{green!25}4       & 11     & 5               \\ \hline
	\multirow{3}{*}{teamform} & 3                  & 12475                     & 15228                    & 2754                  & 2      & \textless{}1    & 2         & \textless{}1     & 972      & 49.0       & 137     & 2      & \textless{}1    \\ 
	& 4                  & 96665                     & 116464                   & 19800                 & 2      & \textless{}1    & 2         & \textless{}1     & 4154     & 34.6       & 9603    & 2      & \textless{}1    \\ 
	& 5                  & 907993                    & 1084752                  & 176760                & 2      & \textless{}1    & 2         & \textless{}1     & \cellcolor{gray!25}         & \cellcolor{gray!25}           & \cellcolor{gray!25}TO       & 2      & \textless{}1    \\ \hline
	\multirow{2}{*}{investor} & 50                 & 211321                    & 673810                   & 29690                     & 441    & 184             & 441      & 249              & \cellcolor{gray!25}         & \cellcolor{gray!25}           & \cellcolor{gray!25}TO       & 364    & \cellcolor{green!25}48              \\ 
	& 100                & 807521                    & 2587510                  & 114390                     & 801    & 3318            & \cellcolor{gray!25}         & \cellcolor{gray!25}OOM                & \cellcolor{gray!25}         & \cellcolor{gray!25}           & \cellcolor{gray!25}TO        & 688    & \cellcolor{green!25}736             \\ \hline
	\multirow{3}{*}{manyECs}   & 500                & 1004                      & 3007                     & 502                   & 6      & 7               & 6        & 7                & \cellcolor{gray!25}         & \cellcolor{gray!25}           & \cellcolor{gray!25}TO       & 5      & \cellcolor{green!25}\textless{}1    \\ 
	& 1000               & 2004                      & 6007                     & 1002                  & 6      & 51              & 6        & 51               & \cellcolor{gray!25}         & \cellcolor{gray!25}           & \cellcolor{gray!25}TO        & 5      & \cellcolor{green!25}\textless{}1    \\ 
	& 5000               & 10004                     & 30007                    & 5002                  & \cellcolor{gray!25}       & \cellcolor{gray!25}SO               & \cellcolor{gray!25}         & \cellcolor{gray!25}SO                & \cellcolor{gray!25}         & \cellcolor{gray!25}           & \cellcolor{gray!25}TO        & 5      & \cellcolor{green!25}\textless{}1    \\
\end{tabular}
}
\end{table}

\paragraph{Discussion}
We observe consistent performance advantage of our algorithm (WP). Even in the mdsm model where the DFL algorithms do not suffer from EC computation (\#EC is just 1), WP's performance is comparable to DFL. The cloud model is where the learning-based approach in~\cite{kelmendi2018value} works well---see  visit\% that are very small.  Our WP performs comparably against DFL\_BRTDP, too. 

The performance advantage of our WP algorithm is eminent, not only in the artificial model of manyECs (where WP is faster by magnitudes), but also in the realistic model investor that comes from a financial application scenario~\cite{McIverM07}. The results for these two models suggest that WP is indeed advantageous when EC computation poses a bottleneck for other algorithms.


Overall, we observe that our WP algorithm can be the first choice when it comes to solving SGs: for some models, it runs much faster than other algorithms; for other models, even if the performances of other algorithms differs a lot, WP's performance is comparable with the best algorithm.

 

\section{Conclusions and Future Work}\label{sec:conclusion}
In this paper, we presented a new BVI algorithm for solving stochastic games. It features global propagation of upper bounds by widest paths, via a novel encoding of the problem to a suitable weighted graph. This way we avoid  computation of end components that often penalizes the performance of the other BVI-based algorithms. Our experimental comparison with known BVI algorithms for SGs demonstrates the efficiency of our algorithm. For correctness of the algorithm, we presented proofs for soundness and convergence.

Extending the current algorithm for more advanced settings is future work---this is much like the results in~\cite{kelmendi2018value} are extended and used in~\cite{AshokKW19,EisentrautKA19preprint,AshokKW19preprint}. In doing so, we hope to make essential use of structures that are unique to those advanced problem settings. Another important direction is to push forward the idea of global propagation in verification and synthesis, seeking further instances of the idea. Finally, pursuing the global propagation idea in the context of reinforcement learning---where problems are often formalized using MDPs and the Bellman operator is heavily utilized---may open up another fruitful collaboration between formal methods and statistical machine learning.



\section*{Acknowledgment} 
The authors are supported by ERATO HASUO Metamathematics for Systems Design Project (No.\ JPMJER1603), JST; I.H.\ is supported by Grant-in-Aid No.\ 15KT0012, JSPS. Thanks are due to Maximilian Weininger and Edon Kelmendi for sharing their implementation, and to Pranav Ashok and David Sprunger for useful discussions and comments. 

%
%
%

\bibliographystyle{splncs04}
\bibliography{ref}

\newpage
\appendix
\section{Step-by-Step Example of Algorithm~\ref{algo:ourAlgorithm}}
We demonstrate the operation of our algorithm (Algorithm~\ref{algo:ourAlgorithm}) step by step. 
Take the SG in Fig.~\ref{fig:ex_sg_ec} as input $
\game$, and let $\varepsilon = 0.01$.

\begin{figure}[tbp]
	\parbox[c]{.3\textwidth}{    \centering
		\scalebox{.6}{    \begin{tikzpicture}[square/.style={regular polygon,regular polygon sides=4}]
			\node[state, square] at (0,0) (s1) {$s_1$};
			\node[state, initial] at (-150:2.5) (s0) {$s_I$};
			\node[state, square] at (0,-2.5) (s2) {$s_2$};
			\node[state, square, minimum size=1cm] at (2.5,0) (target) {$\mathbf{1}$};
			\node[state] at (2.5,-2.5) (sink) {$\mathbf{0}$};
			\draw
			(target) edge[loop right] node[pos=0.3,above right]{$\alpha$} node[pos=0.75,below right]{1} (target)
			(sink)   edge[loop right] node[pos=0.3,above right]{$\alpha$} node[pos=0.75,below right]{1} (sink)
			(s0)     edge[bend left=20] node[pos=0.25,above]{$\alpha$} node[pos=0.7,above]{1} (s1)
			(s1)     edge[bend left=20] node[pos=0.25,above]{$\alpha$} node[pos=0.7,above]{1} (s0)
			(s0)     edge[bend left=20] node[pos=0.25,below]{$\beta$} node[pos=0.7,below]{1} (s2)
			(s2)     edge[bend left=20] node[pos=0.25,below]{$\alpha$} node[pos=0.7,below]{1} (s0)
			(s1)     edge node[pos=0.2,above]{$\beta$} node[pos=0.7,above]{0.8} (target)
			(s2)     edge node[pos=0.2,above]{$\beta$} node[pos=0.7,above]{0.9} (sink);
			\draw[->] (1,0) to node[sloped,above,pos=0.3]{0.2} (sink);
			\draw[->] (1,-2.5) to node[sloped,above,pos=0.3]{0.1} (target);
			\end{tikzpicture}
		}    \caption{An example of SG---$U_{i}$ does not converge to $V(\game)$ in Algorithm~\ref{algo:BVI}}
		\label{fig:ex_sg_ec}
	}
	\hfill
	\parbox[c]{.66\textwidth}{
		\centering
		\subfloat[The WG $\wg_{1}$\label{fig:g1}]{
			\scalebox{0.6}{
				\begin{tikzpicture}[square/.style={regular polygon,regular polygon sides=4}]
				\node[state, square] at (0,0) (s1) {$s_1$};
				\node[state, initial] at (-150:2.5) (s0) {$s_I$};
				\node[state, square] at (0,-2.5) (s2) {$s_2$};
				\node[state, square, minimum size=1cm] at (2.5,0) (target) {$\mathbf{1}$};
				\node[state] at (2.5,-2.5) (sink) {$\mathbf{0}$};
				\draw
				(target) edge[loop right] node[right]{1} (target)
				(sink)   edge[loop right] node[right]{1} (sink)
				(s0)     edge[bend left=20] node[pos=0.25,above]{1} (s1)
				(s1)     edge[bend left=20] node[pos=0.25,above]{1} (s0)
				(s0)     edge[bend left=20] node[pos=0.25,below]{1} (s2)
				(s2)     edge[bend left=20] node[pos=0.25,below]{1} (s0)
				(s1)     edge node[pos=0.2,above]{1} (target)
				(s1)     edge node[pos=0.2,sloped,above]{1} (sink)
				(s2)     edge node[pos=0.2,sloped,above]{1} (target)
				(s2)     edge node[pos=0.2,above]{1} (sink);
				\end{tikzpicture}
			}
		}	
		\subfloat[The WG $\wg_{2}$\label{fig:g2}]{
			\scalebox{0.6}{
				\begin{tikzpicture}[square/.style={regular polygon,regular polygon sides=4}]
				\node[state, square] at (0,0) (s1) {$s_1$};
				\node[state, initial] at (-150:2.5) (s0) {$s_I$};
				\node[state, square] at (0,-2.5) (s2) {$s_2$};
				\node[state, square, minimum size=1cm] at (2.5,0) (target) {$\mathbf{1}$};
				\node[state] at (2.5,-2.5) (sink) {$\mathbf{0}$};
				\draw
				(target) edge[loop right] node[right]{1} (target)
				(sink)   edge[loop right] node[right]{0} (sink)
				(s1)     edge[bend left=20] node[pos=0.25,above left]{1} (s0)
				(s0)     edge[bend left=20] node[pos=0.25,below]{1} (s2)
				(s2)     edge[bend left=20] node[pos=0.25,below]{1} (s0)
				(s1)     edge node[pos=0.2,above]{0.8} (target)
				(s1)     edge node[pos=0.2,sloped,above]{0.8} (sink)
				(s2)     edge node[pos=0.2,sloped,above]{0.1} (target)
				(s2)     edge node[pos=0.2,above]{0.1} (sink);
				\end{tikzpicture}
			}
		}	
		\caption{The WGs  generated from the SG in Fig.~\ref{fig:ex_sg_ec}. States are still distinguished between Maximizer's and Minimizer's; this is for readers' convenience.}
	}
\end{figure}

In iteration $i = 1$ (i.e., in the first iteration) we compute the following.
\begin{itemize}
	\item In Line~\ref{algline:VI} we get $L_1 = \chi_{\{\one\}}$, the characteristic function of $\{\one\}$.
	\item In Line~\ref{algline:playerReduction} we get $\mdp_{1}$ that retains all the actions in $\game$.
	\item In Line~\ref{algline:localPropagation} we get the WG $\wg_{1}$ shown in Fig.~\ref{fig:g1}. 
	\item In Line~\ref{algline:WPP} we get $U_1 = \chi_{S - \{\zero\}}$. 
\end{itemize}

After the first iteration, we have $U_1(s_I) - L_1(s_I) = 1 >\varepsilon$ and the algorithm continues. In the iteration $i=2$ we compute the following.
\begin{itemize}
	\item In Line~\ref{algline:VI}, $L_2(\one) = 1$, $L_2(\zero) = L_2(s_I) = 0$, $L_2(s_1) = 0.8$ and  $L_2(s_2) = 0.1$.
	\item In Line~\ref{algline:playerReduction} we get an MDP $\game^{(2)}$ in which the action $\alpha$ from $s_{I}$ is dropped (cf.\ Fig.~\ref{fig:g2}). This action $\alpha$ is dropped since it was guessed to be suboptimal from the observation $L_{2}(s_{1})>L_{2}(s_{2})$. This happens to be a correct guess, but its correctness does not affect the soundness of the algorithm.
	\item In Line~\ref{algline:localPropagation} we get the WG $\wg_{2}$ depicted in Fig.~\ref{fig:g2}. 
	\item In Line~\ref{algline:WPP}, $U_2(\one) = 1$, $U_2(\one) = 0$, $U_2(s_1) =0.8$ and $U_2(s_I) = U_2(s_2) = 0.1$. 
\end{itemize}
After the second iteration, we have $U_1(s_I) - L_1(s_I) = 0.1 >\varepsilon$ and the algorithm enters yet another iteration. 
In the third iteration, the algorithm applies the Bellman update to $L_2$ which leads to $L_{3}(s_{I})=0.1$. Nothing changes in Lines~\ref{algline:playerReduction}--\ref{algline:WPP} since  we happen to already have $U_2 = V(\game)$. 
The algorithm terminates by witnessing $L_3(s_I) = U_3(s_I) = 0.1$.


\section{The model manyECs}

The model manyECs for the parameter $N=3$ is as follows. Here we omit the sink state $\zero$ and the action labels. 
\begin{center}
	\includegraphics[bb=48 58 895 402,clip,width=\textwidth]{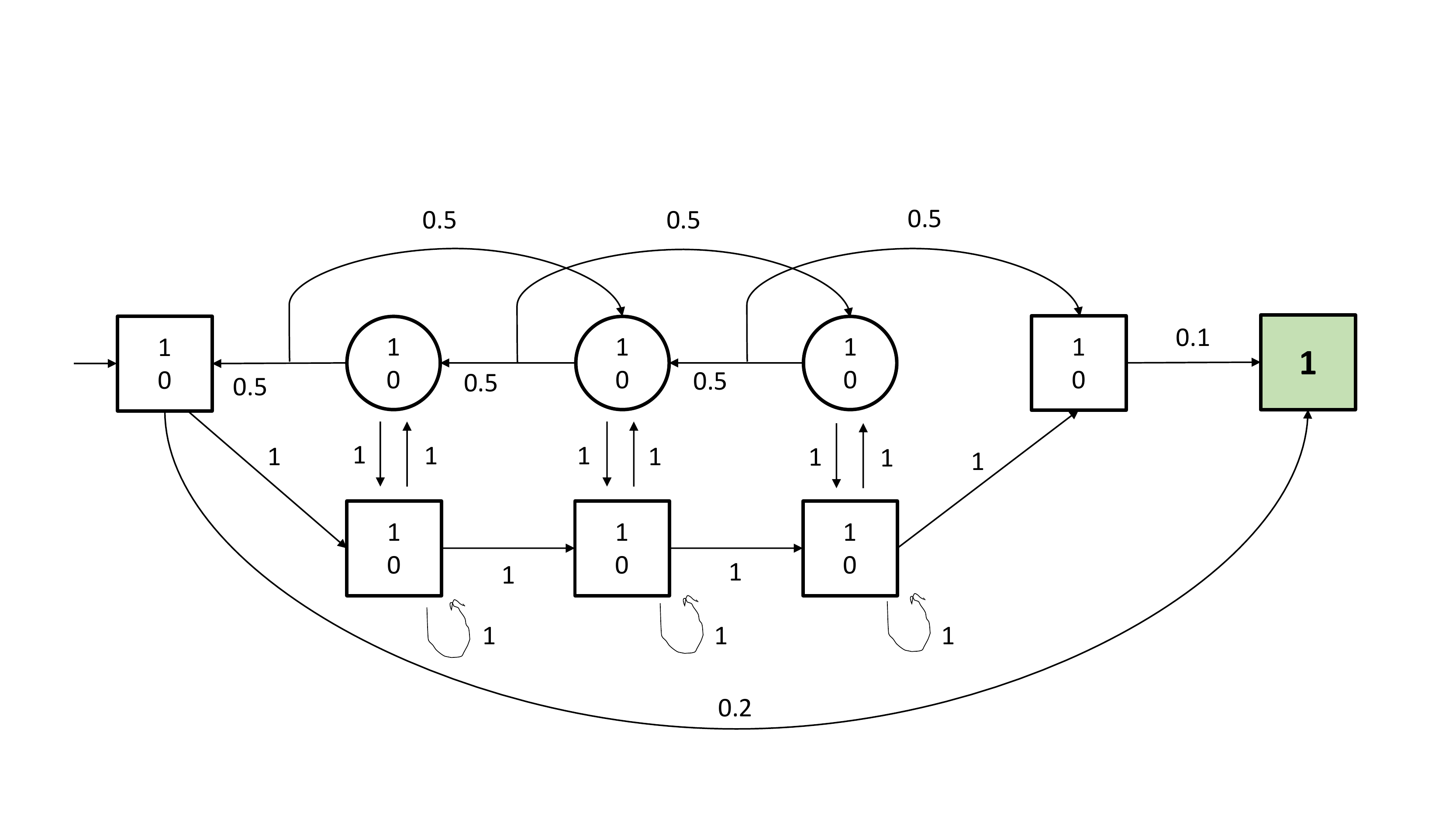}
\end{center}

\end{document}